\documentclass[reqno]{amsart}

\usepackage{hyperref}
\usepackage{graphicx}
\usepackage{curves}
\usepackage{amssymb,amsmath,xcolor,color,framed,epstopdf,amsthm}
\usepackage[english]{babel}
\usepackage{enumerate}
\usepackage{mathrsfs}
\usepackage{color}
\usepackage{empheq}
\usepackage{braket}
\usepackage{hyperref}
\hypersetup{colorlinks = true, linkcolor = {blue} }
\usepackage{cancel}
\usepackage{subcaption}
 \usepackage{lipsum}
 \usepackage{mathtools}
 \usepackage{MnSymbol}
\usepackage{xcolor}

\usepackage[left=3cm, right=3cm, top=3cm, bottom=3cm]{geometry}

\usepackage{tikz} 
\usetikzlibrary{shapes,snakes}
\usepackage{pgf}
\usepackage{pgflibraryarrows}
\usepackage{pgflibrarysnakes}
\usetikzlibrary{decorations.text}
\usepgfmodule{shapes}
\usetikzlibrary{decorations.pathmorphing}
\usetikzlibrary{decorations.markings}
\usetikzlibrary{patterns}
\usetikzlibrary{automata}
\usetikzlibrary {arrows.meta}
\usetikzlibrary{positioning}
\usepackage{pgfplots}
\usetikzlibrary{backgrounds}
\tikzset{partial ellipse/.style args={#1:#2:#3}{insert path={+ (#1:#3) arc (#1:#2:#3)} }}
\tikzset{->-/.style={decoration={ markings, mark=at position #1 with {\arrow{>}}},postaction={decorate}}}
\tikzset{-<-/.style={decoration={ markings, mark=at position #1 with {\arrow{<}}},postaction={decorate}}}

\definecolor{Dgreen}{RGB}{0,153,0}

\newcommand*{\mailto}[1]{\href{mailto:#1}{\nolinkurl{#1}}}

\newtheorem{theorem}{Theorem}

\newtheorem{corollary}[theorem]{Corollary}

\theoremstyle{definition}

\newtheorem{Rhp}{RH problem}

\newcommand{\R}{\mathbb{R}}

\newcommand{\C}{\mathbb{C}}

\newcommand{\be}{\begin{equation}}
\newcommand{\ee}{\end{equation}}
\newcommand{\bea}{\begin{eqnarray}}
\newcommand{\eea}{\end{eqnarray}}

\newcommand{\I}{\mathrm{i}}
\newcommand{\E}{\mathrm{e}}

\newcommand{\lb}{\lambda}

\def\XXint#1#2#3{{\setbox0=\hbox{$#1{#2#3}{\int}$}
     \vcenter{\hbox{$#2#3$}}\kern-.5\wd0}}

\def\d{{\rm d}}

\def\1{\operatorname{Id}}

\def\exp{\operatorname{exp}}

\numberwithin{equation}{section}

\let\oldtocsection=\tocsection
 
\let\oldtocsubsection=\tocsubsection
 
\let\oldtocsubsubsection=\tocsubsubsection
 
\renewcommand{\tocsection}[2]{\hspace{0em}\oldtocsection{#1}{#2}}
\renewcommand{\tocsubsection}[2]{\hspace{2em}\oldtocsubsection{#1}{#2}}
\renewcommand{\tocsubsubsection}[2]{\hspace{4em}\oldtocsubsubsection{#1}{#2}}

\title[Algebro-Geometric solution of the mCH equation]{Riemann-Hilbert approach to the Algebro-Geometric solution of the modified Camassa-Holm equation with linear dispersion term}

\author{Engui Fan}
\address{School of Mathematical Sciences  and Key Laboratory   for Nonlinear Science, Fudan   University, Shanghai 200433, P. R. China.}
\email{faneg@fudna.edu.cn}
\author{Gaozhan Li}
\address{School of Mathematical Sciences  and Key Laboratory   for Nonlinear Science, Fudan   University, Shanghai 200433, P. R. China.}
\email{gzli20@fudna.edu.cn}
\author{Yiling Yang}
\address{College of Mathematics and Statistics, Chongqing University, Chongqing, 401331, P. R. China.}
\email{ylyang@cqu.edu.cn}

\date{}

\begin{document}

\maketitle

\begin{abstract}
This paper aims at providing an exact algebro-geometric solution of the modified Camassa-Holm (mCH) equation derived from hyperelliptic curves in $4(p+q)-1$ genus.
To achieve this goal, we construct the Riemann-Hilbert problems cosponsoring to the mCH equation, which can be solved exactly by the Baker-Akhiezer function. Then the precise expression of the algebro-geometric solution of the mCH equation can be obtained through reconstructed formula.
\end{abstract}

\tableofcontents

\section{Introduction}
In this paper, we consider the modified Camassa-Holm (mCH) equation:
\begin{equation}\label{mCH}
	m_t+\left((u^2-u_x^2)m\right)_x+\omega u_x = 0, \quad m\coloneqq u-u_{xx},
\end{equation} where $u=u(x,t)$ is the function in dimensionless space-time variables $(x,t)$, and $\omega\geq0$ is a  constant characterizing the effect of the linear dispersion.
The mCH equation \eqref{mCH} was first
presented by Fokas \cite{1995On} and Fuchssteiner using recursion operators \cite{1996Some}, and later found by Olver
and Rosenau \cite{Peter1996Tri} via tri-Hamiltonian duality to the bi-Hamiltonian of the mKdV equation (see
also \cite{Hou2017}, referred to as the Fokas-Olver-Rosenau-Qiao equation). It is noticed that the celebrated Camassa-Holm (CH) equation:
\begin{align} \label{ch}
	&m_t+ (um )_x+  u_x m+\omega u_x=0, \quad m=u-u_{x x}
\end{align}	
is the  tri-Hamiltonian duality to the bi-Hamiltonian of the KdV equation \cite{Peter1996Tri}, henceforth, the equation  \eqref{mCH} was referred  to the modified CH equation.

The algebro-geometric solution to integrable nonlinear PDEs, also known as finite-gap potential solutions, was originated from the investigation to the Cauchy problem for the Korteweg-de Vries (KdV) equation with periodic initial conditions. In 1974, Novikov and Dubrovin first linked the Cauchy problem for the KdV equation with periodic initial conditions to algebraic geometry \cite{Dub1974}. In 1975, Its and Matveev found that the finite-gap potential solutions to the KdV equation could be described by the inverse Jacobi problem on a two-sheeted Riemann surface via the spectral theory of periodic Schr\"odinger operators. The potential functions of the continuous spectrum could be represented by theta functions on the Riemann surface, namely the Its-Matveev's formula \cite{Its1975}. Cao, Geng proposed a method to construct algebraic-geometric solutions to integrable equations using the technique of Lax pair nonlinearization \cite{Cao1990,Geng1999}.
Gesztesy and Holden proposed  a systematic
approach to construct algebro-geometric solutions for (1+1) dimensional integrable hierarchy \cite{GesBookI,Ges2003,Ges2005,Ges2008,Ges2017}. Furthermore, this approach has laid the foundation for researching the stability and long-time asymptotic behavior in finite gap potential solutions \cite{Kam2007,Kru2009,Mik2012,Ego2018}. In \cite{Hou2017}, the algebro-geometric solutions to the mCH (namely FORQ) hierarchy \eqref{mCH} with $\omega=0$ are constructed. It is well-known that  the CH equation \eqref{ch} with $\omega\neq0$ can be transformed to the case of  $\omega=0$.  However, the mCH equation \eqref{mCH} with $\omega\neq0$ is a different integrable system of  the case of  $\omega=0$. Until now,  no results to the algebro-geometric solutions to the mCH equation with $\omega\neq0$  have been presented.

To all these algebraic-geometric approaches to the integrable systems, central element  is the so-called Baker-Akhiezer function, which is a meromorphic function on an appropriate Riemann surface introduced  in \cite{GesBookI}. In \cite{Kot-BA}, an approach to construct the Baker-Akhiezer function for nonlinear Schr\"odinger (NLS) equation via Riemann-Hilbert(RH) problem is presented.  It is feasible that one can describe a periodic background solution  by determining a Baker–Akhiezer function analog via appropriate comlplex arc jumps in the RH problem formalism. Several researches  demonstrate the viability of that method and show the power in  studies of integrable systems \cite{Zhao2020,Zhao2023,Kot2018,Kot2019,She2024,Feng2020}.

	
  In the study of long-time asymptotic behavior \cite{Deift1993},  the overall RH problem is frequently  decomposed into solvable model problems at  global and local scales via the Deift-Zhou steepest descent method. The presented paper results offer a  solvable RH problems that serve as global models, along with their corresponding finite-gap potential solution, which will facilitate research into the long-time asymptotic behavior of the mCH equation.
  Moreover, numerous studies \cite{DVZ1994,DZZ16,Gir2021,Buc2007,Iry2013} have demonstrated that finite-gap solutions arise in the asymptotic behavior of various scenarios such as collisionless shock wave transition regions, step-like initial conditions, and soliton gases. 
  In particular,  in the  research to soliton gases for the mCH equation \eqref{mCH}, the  asymptotic leading term in various regions are encompassed of the results of this paper, which  can be  shown as precise finite-gap solutions with concrete coefficients \cite{mCHgas}.
 
In this paper, we show the explicit   algebro-geometric solution for the mCH equation \eqref{mCH} with $\omega\neq0$  by constructing  the corresponding Baker-Akhiezer function via RH problem formalism. 
Based on the inverse scattering transform method, the investigation of solutions  essentially boils down to the study of scattering coefficients. The rest of this paper is organized as follows. In Section \ref{sec:1}, motivated by the previous researches,  we characterize the RH problem corresponding to  the mCH equation. 
We provides a sufficient condition under which an appropriate RH problem can yield an associated real and non-singular solution to the mCH equation \eqref{mCH}. Based on it, we present two RH problems  satisfying the previous stated conditions as examples which corresponding to two cases of explicit solutions.  
In Section  \ref{sec:3}, we  construct a explicitly solvable RH problem, through which we obtain the exactly expression of the finite-gap potential solution  associated with a genus-$(4(p+q)-1)$ hyperelliptic curve for the mCH equation \eqref{mCH} with $\omega\neq0$ in \eqref{AGsol 7}, where $p,q\geq 0$ are integers not all equal to zero.

\section{RH problem framework}\label{sec:1}

In this section, we aim at presenting the construction of a certain RH  problem formulated in the
complex plane with arcs constituting the jump contour,  whose solution can also solves the Lax pair for the mCH equation  \eqref{mCH}.

If we consider the transformation $x\mapsto x, t\mapsto 2t/\omega, u(x,t)\mapsto \sqrt{\omega/2}u(x,2t/\omega)$, the mCH equation \eqref{mCH} becomes 
\begin{equation*}
	m_t+\left((u^2-u_x^2)m\right)_x+2 u_x = 0.
\end{equation*}
So throughout this paper, without loss of generality, we take $\omega=2$.

It is well known that the mCH equation \eqref{mCH} is integrable, arising as the compatibility condition of a Lax pair of linear differential operators \cite{Sch1996}
\begin{equation}
	\Phi_x = X \Phi,\hspace{0.5cm}\Phi_t =T \Phi, \label{lax0}
\end{equation}
where
\begin{equation}
	X=-\frac{\I (\lambda-\lambda^{-1})}{4}\sigma_3+\frac{\I (\lambda+\lambda^{-1}) m}{2}\sigma_2,\nonumber
\end{equation}
\begin{equation}
	T=\frac{\I (\lambda-\lambda^{-1})}{2(\lambda+\lambda^{-1})^2}\sigma_3+\frac{\I (\lambda-\lambda^{-1})}{4} \left(u^{2}-u_{x}^{2}\right)\sigma_3-\I\left(\frac{2\I u-(\lambda-\lambda^{-1}) u_{x}}{2\lambda}+\frac{\lambda+\lambda^{-1}}{2} \left(u^{2}-u_{x}^{2}\right) m \right) \sigma_2.\nonumber
\end{equation}
 The matrices $\sigma_j$, $j=1,2,3$ are the Pauli matrices:
$$
\sigma_1=\begin{pmatrix}
	0&1\\1&0\\
\end{pmatrix},\sigma_2=\begin{pmatrix}
	0&-\mathrm{i}\\\mathrm{i}&0\\
\end{pmatrix},
\sigma_3=\begin{pmatrix}
	1&0\\0&-1\\
\end{pmatrix}.
$$
Noting that \eqref{mCH} is equivalent to the following conservation law type equation
\begin{align*}
	(q)_t+(q(u^2-u_x^2))_x=0
\end{align*}where
\begin{align*}
	q(x,t)=\sqrt{m(x,t)^2+1}.
\end{align*}Let 
\begin{equation}
	y(x,t)=x-\int_{x}^{+\infty} \left(q(s,t)-1\right) ds.\label{transy}
\end{equation}
Under variable $(y,t)$,  the mCH equation \eqref{mCH} and its corresponding  Lax pair  becomes
  \begin{align}\label{y-mCH}
	q_t+2q^2mu_y=0,\quad\quad m=u-q(qu_y)_y.
\end{align}
  and  \begin{align}
	&\Phi_y = \frac{1}{q}X \Phi,\hspace{0.5cm} \Phi_t = \big((u^2-q^2u_y^2)X+T\big) \Phi\label{y-lax}.\end{align} 

In this section, we begin with the RH formalism of the mCH equation \eqref{y-mCH}. The RH problem corresponding to solution to the mCH equation is determined by a piecewise smooth simple curve $\Gamma$ and $2\times2$ matrix valued jump function on $\Gamma$ as follows \cite{Yang2022}.
\begin{Rhp}  \label{RHP0}                       
\hfill
 	\begin{itemize}
 			\item  $M(\lambda):=M(\lb;y,t)$ is holomorphic in $\mathbb{C}\backslash\Gamma$, where $\Gamma$ is a piecewisely smooth oriented curve on the complex plane and admits that if a $\lb\in\Gamma$ then $\pm\lb^{\pm1}$, $\pm \overline{\lambda}^{\pm1}\in\Gamma$. Moreover, $\pm \I,\ 0\notin\Gamma$.
 			\item  $M(\lambda)$ satisfies jump condition \begin{align}\label{jump gas}
 				M_+(\lb)=M_-(\lb)\mathrm{e}^{-\I\theta(\lb)\hat{\sigma}_{3}}J(\lb),
 			\end{align} where
 			\begin{align}
 				\theta(\lb):=\theta(\lambda;y,t)=\frac{1}{4}(\lb-\lb^{-1})(y-8t(\lb+\lb^{-1})^{-2}).\label{def theta}
 			\end{align}
 			In addition, det$J(\lb)=1$ and $J(\lb)$ is required to admits some symmetry such that $M(\lambda)$  satisfies \begin{align}\label{sym0}
 				M(\lambda)=\sigma_1M(-\lambda)\sigma_1=\sigma_2\overline{M(\overline{\lambda})}\sigma_2=M(0)\sigma_3M(-\lambda^{-1})\sigma_3.
 			\end{align}
 				\item $M(\lambda)$ has lower than $-1/2$ singularity on the endpoints of $\Gamma$.
 			\item $M(\lb)\to I$ for $\lb\to\infty$.
 	\end{itemize}
\end{Rhp}

\begin{theorem}\label{DDL}
	For   $M(\lambda)$ satisfies an RH problem \ref{RHP0}, a real and non-singular solution $u(y,t)$ of mCH equation \eqref{y-mCH} is given by the following reconstruction formulae
	\begin{align}\label{recons}
			u(y,t)&=\lim_{\lb\to \I}\frac{1}{\lb-\I}\left(1-\frac{m_1(\lb;y,t)m_2(\lb;y,t)}{m_1(\I;y,t)m_2(\I;y,t)}\right),\qquad
			x(y,t)=y+\log(\frac{m_1(\I;y,t)}{m_2(\I;y,t)}),
	\end{align}where\begin{align*}
		\begin{pmatrix}
			m_1(\lb;y,t)&m_2(\lb;y,t)
		\end{pmatrix}=\begin{pmatrix}
			1&1
		\end{pmatrix}M(\lb;y,t).
	\end{align*}
\end{theorem}
  \begin{proof} 
  	Combining with det$J(\lb)=1$ and $M(\lb)\to I$ for $\lb\to\infty$, we have that det$M(\lb)\equiv1$.
Therefore, the symmetry of $M(\lb)$ in \eqref{sym0}  inspires us to denote
  	\begin{align*}
  		M(0)=\left(\begin{array}{cc}
  			\beta_0  & \eta_0\\
  			\eta_0	& \beta_0
  		\end{array}\right),
  		\hspace{0.3cm}	M(\I)=\left(\begin{array}{cc}
  			f_0  & \frac{\eta_0}{2f_0}\\
  			\frac{\beta_0-1}{\eta_0}f_0	& \frac{\beta_0+1}{2f_0}
  		\end{array}\right),
  		\hspace{0.3cm}\partial_\lambda M(\I)=\left(\begin{array}{cc}
  			\frac{\beta_0-1}{\eta_0}g_1  & g_2\\
  			g_1	& \frac{\beta_0-1}{\eta_0}g_2
  		\end{array}\right),
  	\end{align*}  	
  	where $\beta_0,\ f_0,\ g_1,\ g_2\in\R$, $\eta_0\in\I\R$ and $\beta_0^2-\eta_0^2=1$.
  	Define
\begin{align}
\Psi:=	\Psi(\lb;y,t)=M(\lb;y,t)e^{-\I\theta(\lb;y,t)\sigma_3},\label{psi}
\end{align}
	The definition of $\Psi$ implies that the jump of $\Psi$ is  independent of $y$ and $t$.  Consequently, $\Psi_y\Psi^{-1}$ and $\Psi_t\Psi^{-1}$ have no jump. We analyze $\Psi_y\Psi^{-1}$ first. A directly result from \eqref{psi} is that
\begin{align*}
	\Psi_y\Psi^{-1}=M_yM^{-1}-\frac{\I}{4}\left(\lb-\frac{1}{\lb} \right) M\sigma_3M^{-1},
\end{align*}
which  is a meromorphic function, with possible singularities at $\lb=0$ and $\lb=\infty$. Here, it is noticed that  $M(\lambda)$ has lower than $-1/2$ singularity on the endpoints of $\Gamma$, so $\Psi_y\Psi^{-1}$ does not have singularities at  the endpoints of $\Gamma$. As $\lb\to\infty$,
\begin{align*}
	\Psi_y\Psi^{-1}=-\frac{i}{4}\lb\sigma_3+\frac{\I}{2}\eta\sigma_1+\mathcal{O}(1/\lb),
\end{align*}
while as $\lb\to0$,
\begin{align*}
	\Psi_y\Psi^{-1}=\frac{\I}{4\lb}M(0)\sigma_3M(0)^{-1}+\mathcal{O}(1).
\end{align*}
Therefore, the function
\begin{align*}
	\Psi_y\Psi^{-1}-\frac{\I}{4\lb}M(0)\sigma_3M(0)^{-1}+ \frac{\I}{4}\lb\sigma_3-\frac{\I}{2}\eta\sigma_1
\end{align*}
is  holomorphic in $\mathbb{C}$ and vanish at $\lb\to\infty$. Then, by Liouville's theorem, it vanishes identically, which leads to the result $$ \Psi_y=A\Psi,$$
where
\begin{align*}
	A=&-\frac{\I\lambda}{4}\sigma_3+\frac{\I}{2}\eta\sigma_1+\frac{\I}{4\lambda}(\beta_0^2+\eta_0^2)\sigma_3+\frac{1}{2\lambda}\eta_0\beta_0\sigma_2.
\end{align*}
 Similarly,  	\begin{align*}
	\Psi_t\Psi^{-1}=M_tM^{-1}+2\I\frac{\lb(\lb^2-1)}{(\lb^2+1)^2} M\sigma_3M^{-1},
\end{align*}
is a meromorphic function, with possible singularities at $\lb=\pm i$. From the decomposition
\begin{align*}
	2i\frac{\lb(\lb^2-1)}{(\lb^2+1)^2}=\frac{\I}{\lb+\I}+\frac{\I}{\lb-\I}+\frac{1}{(\lb+\I)^2}-\frac{1}{(\lb-\I)^2},
\end{align*}
we obtain that $$ \Psi_t=B\Psi,$$
where
\begin{align*}
	B=&\frac{2\I(\lambda-\lambda^{-1})}{(\lambda+\lambda^{-1})^2}\left(\begin{array}{cc}
		\beta_0  & -\eta_0\\
		\eta_0	& -\beta_0
	\end{array}\right)
	-\frac{1}{\lambda-\I}\left(\begin{array}{cc}
		2(\frac{\beta_0-1}{\eta_0}g_2f_0+ \frac{\eta_0}{2f_0} g_1)& -2f_0g_2- 2\frac{\beta_0-1}{2f_0} g_1\\
		2\frac{\beta_0-1}{\beta_0+1} g_2f_0+\frac{\beta_0+1}{f_0} g_1	& -2(\frac{\beta_0-1}{\eta_0}g_2f_0+ \frac{\eta_0}{2f_0} g_1)
	\end{array}\right)\\
	&+\frac{1}{\lambda+\I}\left(\begin{array}{cc}
		-2(\frac{\beta_0-1}{\eta_0}g_2f_0+ \frac{\eta_0}{2f_0} g_1)& 2\frac{\beta_0-1}{\beta_0+1} g_2f_0+\frac{\beta_0+1}{f_0} g_1\\
		-2f_0g_2- 2\frac{\beta_0-1}{2f_0} g_1	& 2(\frac{\beta_0-1}{\eta_0}g_2f_0+ \frac{\eta_0}{2f_0} g_1)
	\end{array}\right),
\end{align*} 
 Using the  compatibility condition for the function $\Psi$ results in the compatibility equation
\begin{align*}
	A_t+AB-B_y-BA=0
\end{align*}
yields the mCH equation \eqref{y-mCH}$$	\tilde{q}_t+2\tilde{q}^2\tilde{m}\tilde{u}_y=0,\quad \tilde{q}=\sqrt{1+\tilde{m}^2},\quad \tilde{m}=\tilde{u}-\tilde{q}(\tilde{q}\tilde{u}_y)_y,$$ in the $(y,t)$ variables   via denoting
\begin{align*}
	&\tilde{u}=-\frac{g_1}{f_0}-2(\beta_0+1)f_0g_2,\hspace{0.5cm}\tilde{q}=\frac{1}{\beta_0},\hspace{0.5cm}\tilde{m}=\frac{\eta_0}{\I\beta_0}.
\end{align*} 
Furthermore, let 
\begin{align*}
	\tilde{x}_y=\frac{1}{\tilde{q}},
\end{align*}
above expression   coincides with the formulae \eqref{recons} under $\tilde{x}$ and $\tilde{u}$. Obviously, from its expression, $\tilde{u}$ is real and non-singular since $M(\I)$ is bounded.
\end{proof}

In the present work, we consider some appropriate cases of  $\Gamma,J(\lambda)$ such that  RH problem \ref{RHP0} is precisely solvable. Then the corresponding explicit solutions of the mCH equation \eqref{y-mCH} are derived by the reconstruction formulae \eqref{recons}.

\section{High Genus Algebro-Geometric solution} 
\label{sec:3}

In this chapter, we derive an algebro-geomertic solution $$u^{(AG)}(y,t;\textbf{P}_1,\textbf{P}_2,\textbf{A},\textbf{B}),$$ which satisfies \eqref{y-mCH} and determined by vector-valued parameters \begin{align}\label{para1}
	\textbf{P}_1=\begin{pmatrix}
		c_1&d_1&\cdots&c_p&d_p\\
	\end{pmatrix},\quad	\textbf{P}_2=\begin{pmatrix}
		a_1&b_1&\cdots&a_q&b_q\\
	\end{pmatrix},
\end{align} 
\begin{align}\label{para2}
	\textbf{A}=\begin{pmatrix}
		\alpha_1&\cdots&\alpha_p
	\end{pmatrix},\quad
	\textbf{B}=\begin{pmatrix}
		\beta_1&\cdots&\beta_q
	\end{pmatrix},
\end{align}with integers $p,q\in\mathbb{N}$  not all  zero,  and \begin{align}
	0<c_1<d_1<\cdots<c_p<d_p<\frac{\pi}{2},\quad 0<a_1<b_1<\cdots<a_q<b_q<1,\quad
	\alpha_1,\cdots,\alpha_p,\beta_1,\cdots,\beta_q\neq0.
\end{align}
Regarding the notations above, we replace $\textbf{0}$ with $\textbf{P}_1$ and $\textbf{A}$ ($\textbf{P}_2$ and $\textbf{B}$) when $p=0$ ($q=0$).
The parameters in $\textbf{P}_1$ and  $\textbf{P}_2$ determine a  hyperelliptic curve $\mathcal{R}$ with genus $4(p+q)-1$ defined by $(\lambda,R(\lambda))$, in where $R(\lambda)\sim\lambda^{4(p+q)}$ as $\lambda\to\infty,$ and\small \begin{align*}
	R(\lambda)^2=\prod_{l=1}^{p}(\lambda^4-2\cos(2c_l)\lambda^2+1)(\lambda^4-2\cos(2d_l)\lambda^2+1)\prod_{j=1}^{q}(\lambda^4-(a_j^2+a_j^{-2})\lambda^2+1)(\lambda^4-(b_j^2+b_j^{-2})\lambda^2+1).
\end{align*}\normalsize
Then the  hyperelliptic curve $\mathcal{R}$ is two sheets glued along the branch curve   $\Gamma$, where \begin{align*}
	\Gamma:=&\left\{\lambda\in\C;\left\{\lambda,\bar{\lambda},-\lambda,-\bar{\lambda},\lambda^{-1},\bar{\lambda}^{-1},-\lambda^{-1},-\bar{\lambda}^{-1}\right\}\cap\Gamma^{\dagger}\neq\emptyset\right\},\\
	\Gamma^{\dagger}:=&\bigcup_{l=1}^{p}\left\{\lambda\in\C;|\lambda|=1, \arg\lambda\in[c_l,d_l]\right\}\cup \bigcup_{j=1}^{q}\left\{\lambda\in\I\R;|\lambda|\in[a_j,b_j]\right\}.
\end{align*} 
\begin{figure}[h]
	\begin{minipage}{0.48\textwidth}
		\footnotesize	\begin{tikzpicture} [xscale=1.3,yscale=1.3]
			\draw[dashed] (-2.4,0)--(2.4,0);
			\draw[dashed] (0,-3.5)--(0,3.5);
			\draw[dashed](2,0) arc [start angle=0,end angle=360,radius=2];

			\fill[rotate=-55,white] (-0.2,1.9)--(0.2,1.9)--(0.2,2.1)--(-0.2,2.1); \draw[rotate=-55](0,2)node{\normalcolor$\cdots$};
			\draw[rotate=-30,->-=.6,thick,red] (0,2)node[above]{\normalcolor$\mathrm{e}^{\I d_p}$}arc(90:75:2)node[right]{\normalcolor$\mathrm{e}^{\I c_p}$};
			\draw[rotate=-65,->-=.6,thick,red] (0,2)node[right]{\normalcolor$\mathrm{e}^{\I d_1}$}arc(90:75:2)node[right]{\normalcolor$\mathrm{e}^{\I c_1}$};
			\draw[red] (1.3,1.7)node[right]{$\Gamma_{2p+3q}$};
			\draw[red] (1.9,.6)node[right]{$\Gamma_{3p+3q-1}$};

			\fill[rotate=55,white] (-0.25,1.9)--(0.25,1.9)--(0.25,2.1)--(-0.25,2.1); \draw[rotate=55](0,2)node{\normalcolor$\cdots$};
			\draw[rotate=45,->-=.6,thick,red] (0,2)node[left]{\normalcolor$-\mathrm{e}^{-\I d_p}$}arc(90:75:2)node[above]{\normalcolor$-\mathrm{e}^{-\I c_p}$};
			\draw[rotate=80,->-=.6,thick,red] (0,2)node[left]{\normalcolor$-\mathrm{e}^{-\I d_1}$}arc(90:75:2)node[left]{\normalcolor$-\mathrm{e}^{-\I c_1}$};
			\draw[red] (-1.3,1.7)node[left]{$\Gamma_{q}$};
			\draw[red] (-1.9,.6)node[left]{$\Gamma_{p+q-1}$};
			
			\fill[rotate=-125,white] (-0.25,1.9)--(0.25,1.9)--(0.25,2.1)--(-0.25,2.1); \draw[rotate=-125](0,2)node{\normalcolor$\cdots$};
			\draw[rotate=-135,->-=.6,thick,red] (0,2)node[right]{\normalcolor$\mathrm{e}^{-\I d_p}$}arc(90:75:2)node[below]{\normalcolor$\mathrm{e}^{-\I c_p}$};
			\draw[rotate=-100,->-=.6,thick,red] (0,2)node[right]{\normalcolor$\mathrm{e}^{-\I d_1}$}arc(90:75:2)node[right]{\normalcolor$\mathrm{e}^{-\I c_1}$};
			\draw[red] (1.3,-1.7)node[right]{$\Gamma_{4p+3q-1}$};
			\draw[red] (1.9,-.6)node[right]{$\Gamma_{3p+3q}$};

			\fill[rotate=125,white] (-0.25,1.9)--(0.25,1.9)--(0.25,2.1)--(-0.25,2.1); \draw[rotate=125](0,2)node{$\cdots$};
			\draw[rotate=150,->-=.6,thick,red] (0,2)node[below]{\normalcolor$-\mathrm{e}^{\I d_p}$}arc(90:75:2)node[left]{\normalcolor$-\mathrm{e}^{\I c_p}$};
			\draw[rotate=115,->-=.6,thick,red] (0,2)node[left]{\normalcolor$-\mathrm{e}^{\I d_1}$}arc(90:75:2)node[left]{\normalcolor$-\mathrm{e}^{\I c_1}$};
			\draw[red] (-1.3,-1.7)node[left]{$\Gamma_{2p+q-1}$};
			\draw[red] (-1.9,-.6)node[left]{$\Gamma_{p+q}$};
			
			\draw[->-=.6,thick,red] (0,0.4)node[left]{\normalcolor$\I a_1$}--(0,0.8)node[left]{\normalcolor$\I b_1$};
			\fill[white] (0.1,0.9)--(0.1,1.3)--(-0.1,1.3)--(-0.1,0.9);
			\draw (0,1.1)node{\normalcolor$\cdots$};
			\draw[->-=.6,thick,red] (0,1.4)node[left]{\normalcolor$\I a_q$}--(0,1.8)node[left]{\normalcolor$\I b_q$};
			\draw[red] (0.1,1.6)node[right]{ $\Gamma_{2p+3q-1}$};
			\draw[red] (0.1,.6)node[right]{ $\Gamma_{2p+2q}$};
			
			\draw[-<-=.6,thick,red] (0,-0.4)node[left]{\normalcolor$-\I a_1$}--(0,-0.8)node[left]{\normalcolor$-\I b_1$};
			\fill[white] (0.1,-0.9)--(0.1,-1.3)--(-0.1,-1.3)--(-0.1,-0.9);
			\draw (0,-1.1)node{\normalcolor$\cdots$};
			\draw[-<-=.6,thick,red] (0,-1.4)node[left]{\normalcolor$-\I a_q$}--(0,-1.8)node[left]{\normalcolor$-\I b_q$};
			\draw[red] (0.1,-1.6)node[right]{ $\Gamma_{2p+q}$};
			\draw[red] (0.1,-.6)node[right]{ $\Gamma_{2p+2q-1}$};
			
			\draw[-<-=.6,thick,red] (0,3.4)node[left]{\normalcolor$\I a_1^{-1}$}--(0,3)node[left]{\normalcolor$\I b_1^{-1}$};
			\fill[white] (0.1,2.9)--(0.1,2.7)--(-0.1,2.7)--(-0.1,2.9);
			\draw (0,2.8)node{\normalcolor$\cdots$};
			\draw[-<-=.6,thick,red] (0,2.6)node[right]{\normalcolor$\I a_q^{-1}$}--(0,2.2)node[right]{\normalcolor$\I b_q^{-1}$};
			\draw[red] (0.1,3.2)node[right]{ $\Gamma_{0}$};
			\draw[red] (-0.1,2.4)node[left]{ $\Gamma_{q-1}$};
			
			\draw[->-=.6,thick,red] (0,-3.4)node[left]{\normalcolor$-\I a_1^{-1}$}--(0,-3)node[left]{\normalcolor$-\I b_1^{-1}$};
			\fill[white] (0.1,-2.9)--(0.1,-2.7)--(-0.1,-2.7)--(-0.1,-2.9);
			\draw (0,-2.8)node{\normalcolor$\cdots$};
			\draw[->-=.6,thick,red] (0,-2.6)node[right]{\normalcolor$-\I a_q^{-1}$}--(0,-2.2)node[right]{\normalcolor$-\I b_q^{-1}$};
			\draw[red] (0.1,-3.2)node[right]{ $\Gamma_{4p+4q-1}$};
			\draw[red] (-0.1,-2.4)node[left]{ $\Gamma_{4p+3q}$};

		\end{tikzpicture}\caption*{(a)}
	\end{minipage}
	\begin{minipage}{0.48\textwidth}
		\begin{tikzpicture}[yscale=1.6]
			\draw[white] (-3,0)--(3,0);
			\draw[red,thick](0,2)--(0,1.6);\draw[red] (0.2,1.8)node[right]{$\Gamma_{0}$};
			
			\draw[color={rgb, 255:red, 74; green, 144; blue, 226 },thick,rounded corners,-<-=.5] (0,1.6)--(-.3,1.6)--(-.3,1.2)--(0,1.2);
			\draw[color={rgb, 255:red, 74; green, 144; blue, 226 },thick,densely dashed,rounded corners,-<-=.5] (0,1.2)--(.3,1.2)--(.3,1.6)--(0,1.6);
			\draw[color={rgb, 255:red, 74; green, 144; blue, 226 }] (-.3,1.4) node[left]{$  \mathfrak{a}_1$};
			
			\draw[red,thick](0,1.2)--(0,0.8);\draw[red] (0.2,1)node[right]{$\Gamma_{1}$};
			
			\draw[color={rgb, 255:red, 245; green, 166; blue, 35 },thick,rounded corners,->-=.15]  (-0.2,0.6)--(-0.2,1.4)--(0.2,1.4)--(0.2,0.6)--cycle;
			\draw[color={rgb, 255:red, 245; green, 166; blue, 35 }] (-.2,1)node[left]{$  \mathfrak{b}_1$};
			
			\draw[color={rgb, 255:red, 74; green, 144; blue, 226 },thick,rounded corners,-<-=.5] (0,.8)--(-.3,.8)--(-.3,.4)--(0,.4);
			\draw[color={rgb, 255:red, 74; green, 144; blue, 226 },thick,densely dashed,rounded corners,->-=.5] (0,.8)--(.3,.8)--(.3,.4)--(0,.4);
			\draw[color={rgb, 255:red, 74; green, 144; blue, 226 }] (-.3,.6) node[left]{$  \mathfrak{a}_2$};
			
			\draw[red,thick](0,.4)--(0,0);\draw[red] (0.2,0.2)node[right]{$\Gamma_{2}$};
			\draw (0,-.4)node{$\cdots$};\draw (0,-.8)node{$\cdots$};
			
			\draw[yshift=-90,red,thick](0,2)--(0,1.6);\draw[yshift=-90,red] (0.2,1.8)node[right]{$\Gamma_{4p+4q-3}$};
			
			\draw[yshift=-90,color={rgb, 255:red, 74; green, 144; blue, 226 },thick,rounded corners,-<-=.5] (0,1.6)--(-.3,1.6)--(-.3,1.2)--(0,1.2);
			\draw[yshift=-90,color={rgb, 255:red, 74; green, 144; blue, 226 },thick,densely dashed,rounded corners,-<-=.5] (0,1.2)--(.3,1.2)--(.3,1.6)--(0,1.6);
			\draw[yshift=-90,color={rgb, 255:red, 74; green, 144; blue, 226 }] (-.3,1.4) node[left]{$  \mathfrak{a}_{4p+4q-2}$};
			
			\draw[yshift=-90,red,thick](0,1.2)--(0,0.8);\draw[yshift=-90,red] (0.2,1)node[right]{$\Gamma_{4p+4q-2}$};
			
			\draw[yshift=-90,color={rgb, 255:red, 245; green, 166; blue, 35 },thick,rounded corners,->-=.15]  (-0.2,0.7)--(-0.2,1.3)--(0.2,1.3)--(0.2,0.7)--cycle;
			\draw[yshift=-90,color={rgb, 255:red, 245; green, 166; blue, 35 }] (-.2,1)node[left]{$  \mathfrak{b}_{4p+4q-2}$};
			
			\draw[yshift=-90,color={rgb, 255:red, 74; green, 144; blue, 226 },thick,rounded corners,-<-=.5] (0,.8)--(-.3,.8)--(-.3,.4)--(0,.4);
			\draw[yshift=-90,color={rgb, 255:red, 74; green, 144; blue, 226 },thick,densely dashed,rounded corners,->-=.5] (0,.8)--(.3,.8)--(.3,.4)--(0,.4);
			\draw[yshift=-90,color={rgb, 255:red, 74; green, 144; blue, 226 }] (-.3,.6) node[left]{$  \mathfrak{a}_{4p+4q-1}$};
			
			\draw[yshift=-90,red,thick](0,.4)--(0,0);\draw[yshift=-90,red] (0.2,0.2)node[right]{$\Gamma_{4p+4q-1}$};
			\draw[yshift=-90,color={rgb, 255:red, 245; green, 166; blue, 35 },thick,rounded corners,->-=.15]  (-0.2,-.1)--(-0.2,.5)--(0.2,.5)--(0.2,-.1)--cycle;
			\draw[yshift=-90,color={rgb, 255:red, 245; green, 166; blue, 35 }] (-.2,.2)node[left]{$  \mathfrak{b}_{4p+4q-1}$};
			
		\end{tikzpicture}
		\caption*{(b)}
	\end{minipage}
	\caption{(a) Jump curve $\Gamma$ of RH problem \ref{RHP0 g}; (b) Sketch for  homology basis $\mathfrak{a}_{j}$, $\mathfrak{b}_{j}$ of $\mathcal{R}$.}\label{Gamma curve}
	\normalsize
\end{figure}
  Thus we have $\Gamma=\bigcup_{j=0}^{4p+4q-1}\Gamma_j$ with
\begin{align}
	\Gamma_{j}:=\left\{\begin{array}{lllll}
		\left\{\lambda\in\I\R;-\I\lambda\in[b_{j+1}^{-1},a_{j+1}^{-1}]\right\},\quad &j=0,\cdots,q-1,\\
		\left\{\lambda\in\C;|\lambda|=1, \arg\lambda\in[\pi-d_{q+p-j},\pi-c_{q+p-j}]\right\},\quad &j=q,\cdots,p+q-1, \\
		\left\{\lambda\in\C;|\lambda|=1, \arg\lambda\in	[\pi+c_{j-p-q+1},\pi+d_{j-p-q+1}]\right\},\quad &j=p+q,\cdots,2p+q-1,\\
		\left\{\lambda\in\I\R;-\I\lambda\in[-b_{2p+2q-j},-a_{2p+2q-j}]\right\},\quad &j=2p+q,\cdots,2p+2q-1,\\
		\left\{\lambda\in\I\R;-\I\lambda\in[a_{j-2p-2q+1},b_{j-2p-2q+1}]\right\},\quad &j=2p+2q,\cdots,2p+3q-1,\\
		\left\{\lambda\in\C;|\lambda|=1, \arg\lambda\in[c_{3q+3p-j},d_{3q+3p-j}]\right\}, \quad &j=2p+3q,\cdots,3p+3q-1, \\
		\left\{\lambda\in\C;|\lambda|=1, \arg\lambda\in[-d_{j-3p-3q+1},-c_{j-3p-3q+1}]\right\},\quad &j=3p+3q,\cdots,4p+3q-1,\\
		\left\{\lambda\in\I\R;-\I\lambda\in[-a_{4p+4q-j}^{-1},-b_{4p+4q-j}^{-1}]\right\},&j=4p+3q,\cdots,4p+4q-1,\\	\end{array}\right.
\end{align} which is shown with the orientation of $\Gamma$ in the Figure \ref{Gamma curve} (a).  Basis $\mathfrak{a}_{j}$, $\mathfrak{b}_{j}$ of the first homology group of $\mathcal{R}$ are defined according to the Figure \ref{Gamma curve} (b) following the subscript order of $\Gamma_j$. 
Via these foundations, we generalize   the aforementioned RH problem \ref{RHP0} for the mCH equation \ref{mCH} to the following solvable RH problem \ref{RHP0 g}. Since the same reason it is also an example of the RH problem \ref{RHP0}. Using the Theorem \ref{DDL} and  solvability of this case, we derive an explicit algebro-geometric solution of \eqref{y-mCH} with high genus. 
\begin{Rhp}\label{RHP0 g} 	\hfill
	\begin{itemize}
		\item  $M(\lambda)$ is analytic on $\lambda\in\mathbb{C}\backslash\Gamma$.
		\item $M(\lambda)$ satisfies  jump condition
		\begin{align}
			M_+(\lambda) =M_-(\lambda)\left\{\begin{array}{llll}
				\begin{pmatrix}
					0&-\beta_j^{-1}\mathrm{e}^{-2\I\theta(\lambda)} \\\beta_j\mathrm{e}^{2\I\theta(\lambda)}&0\\
				\end{pmatrix},\quad \lambda\in\Gamma_{j-1}\cup\Gamma_{j+2p+2q-1},j=1,\cdots,q;\\
				\begin{pmatrix}
					0&\I\alpha_j^{-1}\mathrm{e}^{-2\I\theta(\lambda)} \\\I\alpha_j\mathrm{e}^{2\I\theta(\lambda)}&0\\
				\end{pmatrix}, \quad \lambda\in\Gamma_{p+q-j}\cup\Gamma_{3p+3q-j},j=1,\cdots,p;\\
				\begin{pmatrix}
					0&\I\alpha_j\mathrm{e}^{-2\I\theta(\lambda)} \\\I\alpha_j^{-1}\mathrm{e}^{2\I\theta(\lambda)}&0\\
				\end{pmatrix},\quad \lambda\in\Gamma_{j+p+q-1}\cup\Gamma_{j+3p+3q-1},j=1,\cdots,p;\\
				\begin{pmatrix}
					0&-\beta_j\mathrm{e}^{-2\I\theta(\lambda)} \\\beta_j^{-1}\mathrm{e}^{2\I\theta(\lambda)}&0\\
				\end{pmatrix},\quad \lambda\in\Gamma_{2p+2q-j}\cup\Gamma_{4p+4q-j},j=1,\cdots,q.\\	\end{array}\right.
		\end{align}
		\item $M(\lambda)$ has at most $-1/4$ singularity on the endpoints of $\Gamma$.
		\item As $ \lambda\to\infty$, $M(\lambda)=I+\mathcal{O}(\lambda^{-1}).$
	\end{itemize}
\end{Rhp}
The following  Figure \ref{FigR}-\ref{Fig g=11}  presents two examples for $p=q=1$ and $p=2$, $q=1$  with corresponding jump condition for RH problem \ref{RHP0 g}, hyperelliptic curve $\mathcal{R}$ and basis $\mathfrak{a}_{j}$, $\mathfrak{b}_{j}$, $j=1,\cdots,4p+4q-1$.
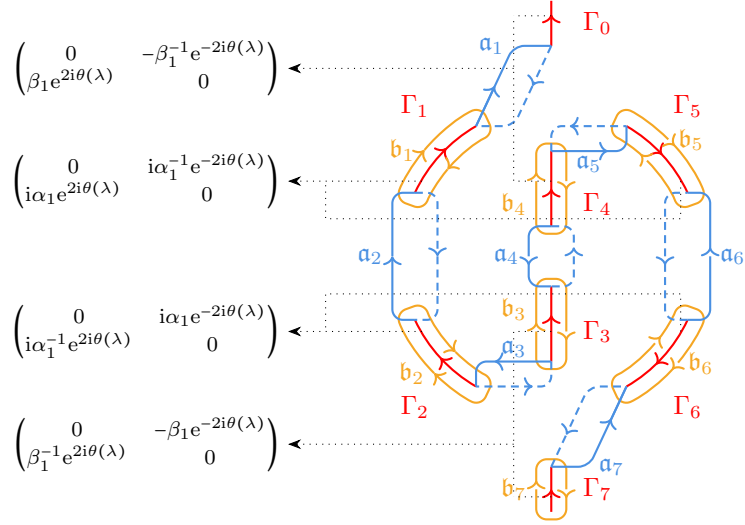
\begin{figure}[h]
	\begin{tikzpicture} [xscale=1,yscale=1]

		
		\draw [red,->-=0.5,thick] (1,1.7320508075688772935274463415059) arc (60:25:2) ;
		
		\draw [color={rgb, 255:red, 245; green, 166; blue, 35 },thick,rounded corners,>-] (1.2728,1.2728)arc(45:65:1.8)--(0.930,1.994) arc (65:45:2.2)node [right]{$ \mathfrak{b}_5$};
		\draw [color={rgb, 255:red, 245; green, 166; blue, 35 },thick,rounded corners,>-](1.556,1.556)arc (45:20:2.2)--(1.691,0.616)arc (20:45:1.8);

		\draw [red](0.3,3.1)node[right]{$\Gamma_0$};
		\draw[red](-1.5,2)node[left]{$\Gamma_1$};
		\draw[red](-1.5,-2)node[left]{$\Gamma_2$};
		\draw 
		[red](0.3,-3.1)node[right]{$\Gamma_7$};
		
		\draw [red](0.3,0.7)node[right]{$\Gamma_4$};
		\draw [red](0.3,-1.)node[right]{$\Gamma_3$};
		\draw[red](1.5,2)node[right]{$\Gamma_5$};
		\draw[red](1.5,-2)node[right]{$\Gamma_6$};
		
		\draw [red,-<-=0.5,thick] (-1,1.7320508075688772935274463415059) arc (120:155:2);
		\draw [color={rgb, 255:red, 245; green, 166; blue, 35 },thick,rotate=95,rounded corners,->] (1.2728,1.2728)arc(45:65:1.8)--(0.930,1.994) arc (65:45:2.2)node [left]{$ \mathfrak{b}_1$};
		\draw [color={rgb, 255:red, 245; green, 166; blue, 35 },thick,rotate=95,rounded corners,->](1.556,1.556)arc (45:20:2.2)--(1.691,0.616)arc (20:45:1.8);
		
		\draw [red,->-=0.5,thick] (-1,-1.7320508075688772935274463415059) arc (240:205:2);
		\draw [color={rgb, 255:red, 245; green, 166; blue, 35 },thick,rotate=180,rounded corners,->] (1.2728,1.2728)arc(45:65:1.8)--(0.930,1.994) arc (65:45:2.2)node [left]{$ \mathfrak{b}_2$};
		\draw [color={rgb, 255:red, 245; green, 166; blue, 35 },thick,rotate=180,rounded corners,->](1.556,1.556)arc (45:20:2.2)--(1.691,0.616)arc (20:45:1.8);
		
		\draw [red,-<-=0.5,thick] (1,-1.7320508075688772935274463415059) arc (-60:-25:2);	
		\draw [color={rgb, 255:red, 245; green, 166; blue, 35 },thick,rotate=-85,rounded corners,>-] (1.2728,1.2728)arc(45:65:1.8)--(0.930,1.994) arc (65:45:2.2)node [right]{$ \mathfrak{b}_6$};
		\draw [color={rgb, 255:red, 245; green, 166; blue, 35 },thick,rotate=-85,rounded corners,>-](1.556,1.556)arc (45:20:2.2)--(1.691,0.616)arc (20:45:1.8);
		
		\draw [red,-<-=0.5,thick] (0,3.4)--(0,2.8);
		\draw [red,->-=0.5,thick] (0,0.4)--(0,1.4) ;
		
		\draw [red,-<-=0.5,thick] (0,-0.4)--(0,-1.4) ;
		\draw [red,->-=0.5,thick] (0,-3.4)--(0,-2.8);
		
		\draw[color={rgb, 255:red, 245; green, 166; blue, 35 },thick,rounded corners,>-]  (-0.2,1)--(-0.2,1.5)--(0.2,1.5)--(0.2,1);\draw[color={rgb, 255:red, 245; green, 166; blue, 35 },thick,rounded corners,>-]  (0.2,1)--(0.2,0.3)--(-0.2,0.3)--(-0.2,1);
		\draw[color={rgb, 255:red, 245; green, 166; blue, 35 },thick,](-0.2,0.7)node[left]{$ \mathfrak{b}_4$};
		\draw[color={rgb, 255:red, 245; green, 166; blue, 35 },thick,rounded corners,-<]  (-0.2,-1)--(-0.2,-1.5)--(0.2,-1.5)--(0.2,-1);\draw[color={rgb, 255:red, 245; green, 166; blue, 35 },thick,rounded corners,-<]  (0.2,-1)--(0.2,-0.3)--(-0.2,-0.3)--(-0.2,-1);\draw[color={rgb, 255:red, 245; green, 166; blue, 35 },thick,](-0.2,-0.7)node[left]{$ \mathfrak{b}_3$};
		\draw[color={rgb, 255:red, 245; green, 166; blue, 35 },thick,rounded corners,-<]  (-0.2,-3.1)--(-0.2,-3.5)--(0.2,-3.5)--(0.2,-3.1);\draw[color={rgb, 255:red, 245; green, 166; blue, 35 },thick,rounded corners,-<]  (0.2,-3.1)--(0.2,-2.7)--(-0.2,-2.7)--(-0.2,-3.1)node[left]{$ \mathfrak{b}_7$};
		
		\draw[color={rgb, 255:red, 74; green, 144; blue, 226 },thick,rounded corners] (-1,1.732) --(-0.5,2.8)node[left]{$  \mathfrak{a}_1$}--(0,2.8) ;
		\draw [color={rgb, 255:red, 74; green, 144; blue, 226 },thick,rounded corners,<-](-0.75,2.266) --(-1,1.732);
		\draw[color={rgb, 255:red, 74; green, 144; blue, 226 },thick,rounded corners,densely dashed,->] (0,2.8) --(-0.25,2.266);
		\draw [color={rgb, 255:red, 74; green, 144; blue, 226 },thick,rounded 	corners,densely dashed](-0.25,2.266)--(-0.5,1.732)--(-1,1.732);
		
		\draw[color={rgb, 255:red, 74; green, 144; blue, 226 },thick,rotate=180,rounded corners] (-1,1.732) --(-0.5,2.8)node[right]{$  \mathfrak{a}_7$}--(0,2.8) ;
		\draw [color={rgb, 255:red, 74; green, 144; blue, 226 },thick,rotate=180,rounded corners,-<](-1,1.732)--(-0.75,2.266) ;
		\draw[color={rgb, 255:red, 74; green, 144; blue, 226 },thick,rotate=180,rounded corners,densely dashed,-<] (0,2.8) --(-0.25,2.266);
		\draw [color={rgb, 255:red, 74; green, 144; blue, 226 },thick,rotate=180,rounded 	corners,densely dashed](-0.25,2.266)--(-0.5,1.732)--(-1,1.732);
		
		\draw[color={rgb, 255:red, 74; green, 144; blue, 226 },thick,rounded corners] (-1.8126,0.8452)--(-2.1126,0.8452)--(-2.1126,0) node[left]{$  \mathfrak{a}_2$};
		\draw [color={rgb, 255:red, 74; green, 144; blue, 226 },thick,rounded corners,<-](-2.1126,0)--(-2.1126,-0.8452)--(-1.8126,-0.8452);
		\draw[color={rgb, 255:red, 74; green, 144; blue, 226 },thick,rounded corners,densely dashed,->] (-1.8126,0.8452)--(-1.5126,0.8452)--(-1.5126,0);
		\draw [color={rgb, 255:red, 74; green, 144; blue, 226 },thick,rounded 	corners,densely dashed](-1.5126,0)--(-1.5126,-0.8452)--(-1.8126,-0.8452);
		
		\draw[color={rgb, 255:red, 74; green, 144; blue, 226 },thick,rotate=180,rounded corners] (-1.8126,0.8452)--(-2.1126,0.8452)--(-2.1126,0) node[right]{$  \mathfrak{a}_6$};
		\draw [color={rgb, 255:red, 74; green, 144; blue, 226 },thick,rotate=180,rounded corners,>-](-2.1126,0)--(-2.1126,-0.8452)--(-1.8126,-0.8452);
		\draw[color={rgb, 255:red, 74; green, 144; blue, 226 },thick,rotate=180,rounded corners,densely dashed,-<] (-1.8126,0.8452)--(-1.5126,0.8452)--(-1.5126,0);
		\draw [color={rgb, 255:red, 74; green, 144; blue, 226 },thick,rotate=180,rounded 	corners,densely dashed](-1.5126,0)--(-1.5126,-0.8452)--(-1.8126,-0.8452);
		
		\draw[color={rgb, 255:red, 74; green, 144; blue, 226 },thick,rounded corners] (0,0.4)--(-0.3,0.4)--(-0.3,0) node[left]{$  \mathfrak{a}_4$};
		\draw [color={rgb, 255:red, 74; green, 144; blue, 226 },thick,rounded corners,>-](-0.3,0)--(-0.3,-0.4)--(0,-0.4);
		\draw[color={rgb, 255:red, 74; green, 144; blue, 226 },thick,rounded corners,densely dashed,-<] (0,0.4)--(0.3,0.4)--(0.3,0);
		\draw [color={rgb, 255:red, 74; green, 144; blue, 226 },thick,rounded 	corners,densely dashed](0.3,0)--(0.3,-0.4)--(0,-0.4);
		
		\draw[color={rgb, 255:red, 74; green, 144; blue, 226 },thick,rounded corners,-<-=0.5] (-1,-1.732) --(-1,-1.4)--(0,-1.4);
		\draw [color={rgb, 255:red, 74; green, 144; blue, 226 },thick,rounded 	corners,densely dashed,-<-=0.5](-0.,-1.4)--(-0,-1.732)--(-1,-1.732) (-0.2,-1.232)node[left]{$  \mathfrak{a}_3$};

		\draw[color={rgb, 255:red, 74; green, 144; blue, 226 },thick,rotate=180,rounded corners,-<-=0.5] (-1,-1.732) --(-1,-1.4)--(0,-1.4);
		\draw [color={rgb, 255:red, 74; green, 144; blue, 226 },thick,rotate=180, rounded 	corners,densely dashed,-<-=0.5](-0.,-1.4)--(-0,-1.732)--(-1,-1.732) (-0.2,-1.232)node[right]{$  \mathfrak{a}_5$};
		
		\draw[-{Stealth[length=5]}, dotted] (-1.72,1)--(-3.5,1)node[left] {\footnotesize$\begin{pmatrix}
				0&\I\alpha_1^{-1}\mathrm{e}^{-2\I\theta(\lambda)} \\\I\alpha_1\mathrm{e}^{2\I\theta(\lambda)}&0\\
			\end{pmatrix}$};  
		\draw[ -, dotted] (1.72,1)-- (1.72,0.5)--(-3,0.5)--(-3,1);
		
		\draw[-{Stealth[length=5]}, dotted] (-1.72,-1)--(-3.5,-1)node[left] {\footnotesize$\begin{pmatrix}
				0&\I\alpha_1\mathrm{e}^{-2\I\theta(\lambda)} \\\I\alpha_1^{-1}\mathrm{e}^{2\I\theta(\lambda)}&0\\
			\end{pmatrix}$};  
		\draw[-, dotted] (1.72,-1)-- (1.72,-0.5)--(-3,-0.5)--(-3,-1);

		\draw[-, dotted] (0,3.2)-- (-.5,3.2);
		\draw [-{Stealth[length=5]}, dotted] (-.5,2.5)--(-3.5,2.5)node[left] {\footnotesize$\begin{pmatrix}
				0&-\beta_1^{-1}\mathrm{e}^{-2\I\theta(\lambda)} \\\beta_1\mathrm{e}^{2\I\theta(\lambda)}&0\\
			\end{pmatrix}$};
		\draw[ dotted] (0,1)--(-0.5,1)--(-0.5,3.2);
		
		\draw[-, dotted] (0,-3.2)-- (-.5,-3.2);
		\draw[-{Stealth[length=5]}, dotted] (-.5,-2.5)--(-3.5,-2.5)node[left] {\footnotesize$\begin{pmatrix}
				0&-\beta_1\mathrm{e}^{-2\I\theta(\lambda)} \\\beta_1^{-1}\mathrm{e}^{2\I\theta(\lambda)}&0\\
			\end{pmatrix}$}; 
		\draw[ dotted] (0,-1)--(-0.5,-1)--(-0.5,-3.2);
		
	\end{tikzpicture}
	\caption{ Jump curve  of the RH problem \ref{RHP0 g},  $\mathcal{R}$  and  basis of  its first homology group for $p=q=1$ case.} \label{FigR}
	
\end{figure}
	\begin{figure}[h]
	\begin{tikzpicture} [xscale=1.3,yscale=1.3]

		%
		
		
		\draw [color={rgb, 255:red, 245; green, 166; blue, 35 },thick,rotate=-10,red,->-=0.5,thick](1.414,1.414)arc(45:25:2);
		\draw [color={rgb, 255:red, 245; green, 166; blue, 35 },thick,rotate=-7.5,rounded corners,->-=.2](1.556,1.556)arc (45:20:2.2)--(1.691,0.616)arc (20:45:1.8)--cycle;
		\draw [color={rgb, 255:red, 245; green, 166; blue, 35 }](1.95,1)node[right]{$\mathfrak{b}_{8}$};

		\draw [color={rgb, 255:red, 245; green, 166; blue, 35 },thick,rotate=30,red,->-=0.5,thick](1.414,1.414)arc(45:25:2);
		\draw [color={rgb, 255:red, 245; green, 166; blue, 35 },thick,rotate=32.5,rounded corners,->-=.2](1.556,1.556)arc (45:20:2.2)--(1.691,0.616)arc (20:45:1.8)--cycle;
		\draw[color={rgb, 255:red, 245; green, 166; blue, 35 }](1,2) node[above]{$\mathfrak{b}_{7}$};
		
		\draw [red,thick,rotate=80,->-=0.5,thick](1.414,1.414)arc(45:25:2);
		\draw [color={rgb, 255:red, 245; green, 166; blue, 35 },thick,rotate=82.5,rounded corners,->-=.2](1.556,1.556)arc (45:20:2.2)--(1.691,0.616)arc (20:45:1.8)--cycle;
		\draw[color={rgb, 255:red, 245; green, 166; blue, 35 }](-1,2) node[above]{$\mathfrak{b}_{1}$};
		
		\draw [red,thick,rotate=120,->-=0.5,thick](1.414,1.414)arc(45:25:2);
		\draw [color={rgb, 255:red, 245; green, 166; blue, 35 },thick,rotate=122.5,rounded corners,->-=.2](1.556,1.556)arc (45:20:2.2)--(1.691,0.616)arc (20:45:1.8)--cycle; 
		\draw [color={rgb, 255:red, 245; green, 166; blue, 35 }](-1.95,1)node[left]{$\mathfrak{b}_{2}$};

		\draw [color={rgb, 255:red, 245; green, 166; blue, 35 },thick,rotate=170,red,->-=0.5,thick](1.414,1.414)arc(45:25:2);
		\draw [color={rgb, 255:red, 245; green, 166; blue, 35 },thick,rotate=172.5,rounded corners,->-=.2](1.556,1.556)arc (45:20:2.2)--(1.691,0.616)arc (20:45:1.8)--cycle;
		\draw [color={rgb, 255:red, 245; green, 166; blue, 35 }](-1.95,-1)node[left]{$\mathfrak{b}_{3}$};
		
		\draw [color={rgb, 255:red, 245; green, 166; blue, 35 },thick,rotate=210,red,->-=0.5,thick](1.414,1.414)arc(45:25:2);
		\draw [color={rgb, 255:red, 245; green, 166; blue, 35 },thick,rotate=212.5,rounded corners,->-=.2](1.556,1.556)arc (45:20:2.2)--(1.691,0.616)arc (20:45:1.8)--cycle;
		\draw [color={rgb, 255:red, 245; green, 166; blue, 35 }](-1,-2)node[below]{$\mathfrak{b}_{4}$};
		
		\draw [color={rgb, 255:red, 245; green, 166; blue, 35 },thick,rotate=260,red,->-=0.5,thick](1.414,1.414)arc(45:25:2);
		\draw [color={rgb, 255:red, 245; green, 166; blue, 35 },thick,rotate=262.5,rounded corners,->-=.2](1.556,1.556)arc (45:20:2.2)--(1.691,0.616)arc (20:45:1.8)--cycle;
		\draw [color={rgb, 255:red, 245; green, 166; blue, 35 }](1,-2)node[below]{$\mathfrak{b}_{10}$};
		
		\draw [color={rgb, 255:red, 245; green, 166; blue, 35 },thick,rotate=300,red,->-=0.5,thick](1.414,1.414)arc(45:25:2);
		\draw [color={rgb, 255:red, 245; green, 166; blue, 35 },thick,rotate=302.5,rounded corners,->-=.2](1.556,1.556)arc (45:20:2.2)--(1.691,0.616)arc (20:45:1.8)--cycle;
		\draw [color={rgb, 255:red, 245; green, 166; blue, 35 }](1.95,-1)node[right]{$\mathfrak{b}_{9}$};

		\draw [red](0.3,3.1)node[right]{$\Gamma_0$};
		\draw[red](-1,2.1)node[left]{$\Gamma_1$};
		\draw[red](-1.1,-2.1)node[left]{$\Gamma_4$};
		\draw 
		[red](0.3,-3.1)node[right]{$\Gamma_7$};
		\draw[red](-1.8,1.3)node[left]{$\Gamma_2$};
		\draw[red](-1.8,-1.4)node[left]{$\Gamma_3$};

		\draw[red](1.1,2.1)node[right]{$\Gamma_7$};
		\draw[red](1.1,-2.1)node[right]{$\Gamma_{10}$};
		
		\draw[red](1.8,1.4)node[right]{$\Gamma_8$};
		\draw[red](1.8,-1.4)node[right]{$\Gamma_{9}$};

		\draw [red](0.3,0.7)node[right]{$\Gamma_6$};
		\draw [red](0.3,-1.)node[right]{$\Gamma_5$};

		\draw [red,-<-=0.5,thick] (0,3.4)--(0,2.8);
		\draw [red,->-=0.5,thick] (0,0.4)--(0,1.4) ;
		
		\draw [red,-<-=0.5,thick] (0,-0.4)--(0,-1.4) ;
		\draw [red,->-=0.5,thick] (0,-3.4)--(0,-2.8);
		
		\draw[color={rgb, 255:red, 245; green, 166; blue, 35 },thick,rounded corners,>-]  (-0.2,1)--(-0.2,1.5)--(0.2,1.5)--(0.2,1);\draw[color={rgb, 255:red, 245; green, 166; blue, 35 },thick,rounded corners,>-]  (0.2,1)--(0.2,0.3)--(-0.2,0.3)--(-0.2,1)node[left]{$ \mathfrak{b}_6$};
		\draw[color={rgb, 255:red, 245; green, 166; blue, 35 },thick,rounded corners,-<]  (-0.2,-1)--(-0.2,-1.5)--(0.2,-1.5)--(0.2,-1);\draw[color={rgb, 255:red, 245; green, 166; blue, 35 },thick,rounded corners,-<]  (0.2,-1)--(0.2,-0.3)--(-0.2,-0.3)--(-0.2,-1)node[left]{$ \mathfrak{b}_5$};
		\draw[color={rgb, 255:red, 245; green, 166; blue, 35 },thick,rounded corners,-<]  (-0.2,-3.1)--(-0.2,-3.5)--(0.2,-3.5)--(0.2,-3.1);\draw[color={rgb, 255:red, 245; green, 166; blue, 35 },thick,rounded corners,-<]  (0.2,-3.1)--(0.2,-2.7)--(-0.2,-2.7)--(-0.2,-3.1)node[left]{$ \mathfrak{b}_7$};
		
		\draw[color={rgb, 255:red, 74; green, 144; blue, 226 },thick,rounded corners,->-=0.5] (-.52,1.93) --(-0.3,2.8)--(0,2.8) ;
		\draw [color={rgb, 255:red, 74; green, 144; blue, 226 },thick,rounded 	corners,densely dashed,->-=.5](0,2.8)--(-.2,1.93)--(-0.52,1.93);
		\draw[color={rgb, 255:red, 74; green, 144; blue, 226 }] (-.5,2.7) node[left]{$  \mathfrak{a}_1$};
		
		\draw[color={rgb, 255:red, 74; green, 144; blue, 226 },thick,rotate=180,rounded corners,-<-=0.5] (-.52,1.93) --(-0.3,2.8)--(0,2.8) ;
		\draw [color={rgb, 255:red, 74; green, 144; blue, 226 },thick,rotate=180,rounded 	corners,densely dashed,-<-=.5](0,2.8)--(-.2,1.93)--(-0.52,1.93);\draw[color={rgb, 255:red, 74; green, 144; blue, 226 }] (.5,-2.7) node[right]{$  \mathfrak{a}_{11}$};
		
		\draw[color={rgb, 255:red, 74; green, 144; blue, 226 },thick,rounded corners,-<-=.5] (-1.93,0.52)--(-2.23,0.52)--(-2.23,-0.52)--(-1.93,-0.52);
		\draw[color={rgb, 255:red, 74; green, 144; blue, 226 },thick,densely dashed,rounded corners,->-=.5] (-1.93,0.52)--(-1.63,0.52)--(-1.63,-0.52)--(-1.93,-0.52);
		\draw[color={rgb, 255:red, 74; green, 144; blue, 226 }] (-2.3,0) node[left]{$  \mathfrak{a}_3$};
		
		\draw[color={rgb, 255:red, 74; green, 144; blue, 226 },thick,rotate=180,rounded corners,->-=.5] (-1.93,0.52)--(-2.23,0.52)--(-2.23,-0.52)--(-1.93,-0.52);
		\draw[color={rgb, 255:red, 74; green, 144; blue, 226 },thick,rotate=180,densely dashed,rounded corners,-<-=.5] (-1.93,0.52)--(-1.63,0.52)--(-1.63,-0.52)--(-1.93,-0.52);
		\draw[color={rgb, 255:red, 74; green, 144; blue, 226 }] (2.3,0) node[right]{$  \mathfrak{a}_9$};
		
		\draw[color={rgb, 255:red, 74; green, 144; blue, 226 },thick,rotate=-45,rounded corners,-<-=.5] (-1.97,0.35)--(-2.27,0.35)--(-2.27,-0.35)--(-1.97,-0.35);
		\draw[color={rgb, 255:red, 74; green, 144; blue, 226 },thick,rotate=-45,densely dashed,rounded corners,->-=.5] (-1.97,0.35)--(-1.67,0.35)--(-1.67,-0.35)--(-1.97,-0.35);
		\draw[color={rgb, 255:red, 74; green, 144; blue, 226 }] (-1.8,1.8) node[left]{$  \mathfrak{a}_2$};
		
		\draw[color={rgb, 255:red, 74; green, 144; blue, 226 },thick,rotate=45,rounded corners,-<-=.5] (-1.97,0.35)--(-2.27,0.35)--(-2.27,-0.35)--(-1.97,-0.35);
		\draw[color={rgb, 255:red, 74; green, 144; blue, 226 },thick,rotate=45,densely dashed,rounded corners,->-=.5] (-1.97,0.35)--(-1.67,0.35)--(-1.67,-0.35)--(-1.97,-0.35);
		\draw[color={rgb, 255:red, 74; green, 144; blue, 226 }] (-1.8,-1.8) node[left]{$  \mathfrak{a}_4$};
		
		\draw[color={rgb, 255:red, 74; green, 144; blue, 226 },thick,rotate=135,rounded corners,->-=.5] (-1.97,0.35)--(-2.27,0.35)--(-2.27,-0.35)--(-1.97,-0.35);
		\draw[color={rgb, 255:red, 74; green, 144; blue, 226 },thick,rotate=135,densely dashed,rounded corners,-<-=.5] (-1.97,0.35)--(-1.67,0.35)--(-1.67,-0.35)--(-1.97,-0.35);
		\draw[color={rgb, 255:red, 74; green, 144; blue, 226 }] (1.8,-1.8) node[right]{$  \mathfrak{a}_{10}$};
		
		\draw[color={rgb, 255:red, 74; green, 144; blue, 226 },thick,rotate=225,rounded corners,->-=.5] (-1.97,0.35)--(-2.27,0.35)--(-2.27,-0.35)--(-1.97,-0.35);
		\draw[color={rgb, 255:red, 74; green, 144; blue, 226 },thick,rotate=225,densely dashed,rounded corners,-<-=.5] (-1.97,0.35)--(-1.67,0.35)--(-1.67,-0.35)--(-1.97,-0.35);
		\draw[color={rgb, 255:red, 74; green, 144; blue, 226 }] (1.8,1.8) node[right]{$  \mathfrak{a}_8$};
		
		\draw[color={rgb, 255:red, 74; green, 144; blue, 226 },thick,rounded corners] (0,0.4)--(-0.3,0.4)--(-0.3,0) node[left]{$  \mathfrak{a}_6$};
		\draw [color={rgb, 255:red, 74; green, 144; blue, 226 },thick,rounded corners,>-](-0.3,0)--(-0.3,-0.4)--(0,-0.4);
		\draw[color={rgb, 255:red, 74; green, 144; blue, 226 },thick,rounded corners,densely dashed,-<] (0,0.4)--(0.3,0.4)--(0.3,0);
		\draw [color={rgb, 255:red, 74; green, 144; blue, 226 },thick,rounded 	corners,densely dashed](0.3,0)--(0.3,-0.4)--(0,-0.4);
		
		\draw[color={rgb, 255:red, 74; green, 144; blue, 226 },thick,rounded corners,-<-=0.3] (-0.52,-1.93) --(-.52,-1.4)--(0,-1.4);
		\draw [color={rgb, 255:red, 74; green, 144; blue, 226 },thick,rounded 	corners,densely dashed,-<-=0.3](-0.,-1.4)--(-0,-1.93)--(-0.52,-1.93);
		\draw[color={rgb, 255:red, 74; green, 144; blue, 226 }] (-.4,-1.3) node[left]{$  \mathfrak{a}_5$};
		
		\draw[color={rgb, 255:red, 74; green, 144; blue, 226 },thick,rotate=180,rounded corners,->-=0.3] (-0.52,-1.93) --(-.52,-1.4)--(0,-1.4);
		\draw [color={rgb, 255:red, 74; green, 144; blue, 226 },thick,rotate=180,rounded 	corners,densely dashed,->-=0.3](-0.,-1.4)--(-0,-1.93) --(-0.52,-1.93);
		
		\draw[color={rgb, 255:red, 74; green, 144; blue, 226 }] (.4,1.3) node[right]{$  \mathfrak{a}_7$};
		\draw[color={rgb, 255:red, 74; green, 144; blue, 226 }] (.5,-2.7) node[right]{$  \mathfrak{a}_{11}$};

		\draw[-{Stealth[length=5]}, dotted]  (-1.82,.7)--(-1.82,.3)--(-3.5,.3)node[left] {\footnotesize$\begin{pmatrix}
				0&\I\alpha_1^{-1}\mathrm{e}^{-2\I\theta(\lambda)} \\\I\alpha_1\mathrm{e}^{2\I\theta(\lambda)}&0\\
			\end{pmatrix}$};  
		\draw[ -, dotted] (1.82,.7)-- (1.82,0.3)--(-3,0.3);
		
		\draw[-{Stealth[length=5]}, dotted] (-1.82,-.7)--(-1.82,-.3)--(-3.5,-.3)node[left] {\footnotesize$\begin{pmatrix}
				0&\I\alpha_1\mathrm{e}^{-2\I\theta(\lambda)} \\\I\alpha_1^{-1}\mathrm{e}^{2\I\theta(\lambda)}&0\\
			\end{pmatrix}$};  
		\draw[-, dotted] (1.82,-.7)-- (1.82,-0.3)--(-3,-0.3);
		
		\draw[-{Stealth[length=5]}, dotted]  (-.82,1.7)--(-.82,1.5)--(-3.5,1.5)node[left] {\footnotesize$\begin{pmatrix}
				0&\I\alpha_2^{-1}\mathrm{e}^{-2\I\theta(\lambda)} \\\I\alpha_2\mathrm{e}^{2\I\theta(\lambda)}&0\\
			\end{pmatrix}$};  
		\draw[ -, dotted] (.82,1.7)-- (.82,1.5)--(-3,1.5);
		
		\draw[-{Stealth[length=5]}, dotted]  (-.82,-1.7)--(-.82,-1.5)--(-3.5,-1.5)node[left] {\footnotesize$\begin{pmatrix}
				0&\I\alpha_2\mathrm{e}^{-2\I\theta(\lambda)} \\\I\alpha_2^{-1}\mathrm{e}^{2\I\theta(\lambda)}&0\\
			\end{pmatrix}$};  
		\draw[ -, dotted] (.82,-1.7)-- (.82,-1.5)--(-3,-1.5);

		\draw[-, dotted] (0,3.2)-- (-.5,3.2);
		\draw [-{Stealth[length=5]}, dotted] (-.5,2.5)--(-3.5,2.5)node[left] {\footnotesize$\begin{pmatrix}
				0&-\beta_1^{-1}\mathrm{e}^{-2\I\theta(\lambda)} \\\beta_1\mathrm{e}^{2\I\theta(\lambda)}&0\\
			\end{pmatrix}$};
		\draw[ dotted] (0,1)--(-0.5,1)--(-0.5,3.2);
		
		\draw[-, dotted] (0,-3.2)-- (-.5,-3.2);
		\draw[-{Stealth[length=5]}, dotted] (-.5,-2.5)--(-3.5,-2.5)node[left] {\footnotesize$\begin{pmatrix}
				0&-\beta_1\mathrm{e}^{-2\I\theta(\lambda)} \\\beta_1^{-1}\mathrm{e}^{2\I\theta(\lambda)}&0\\
			\end{pmatrix}$}; 
		\draw[ dotted] (0,-1)--(-0.5,-1)--(-0.5,-3.2);
		
	\end{tikzpicture}
	\caption{ Jump curve $\Gamma=\bigcup_{j=0}^{4p+4q-1}\Gamma_j$ of the RH problem \ref{RHP0 g},  $\mathcal{R}$  and basis of  its  first homology group for $p=2$, $q=1$ case.} \label{Fig g=11}
	
\end{figure}
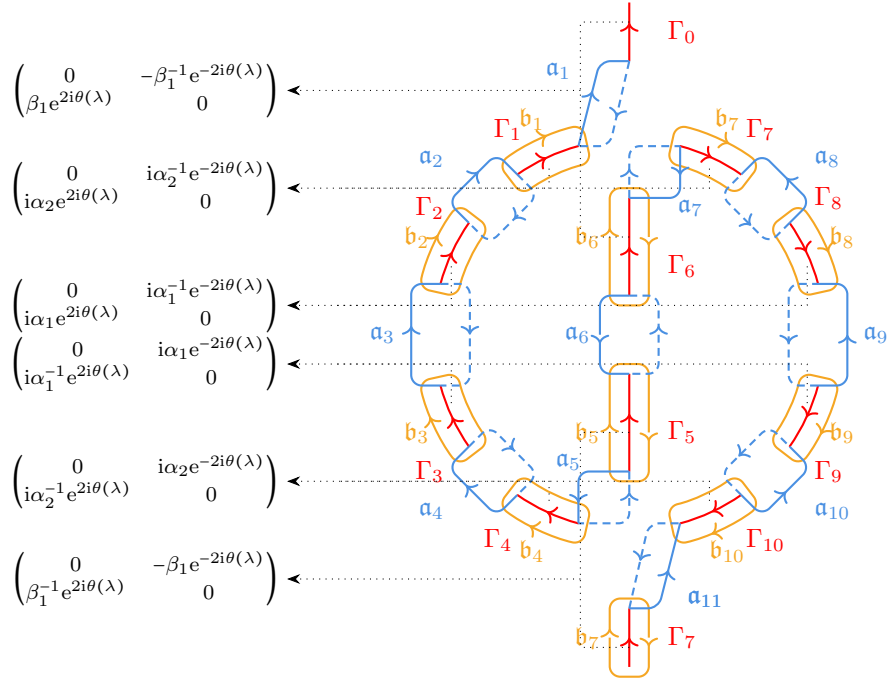

\subsection{$g-$function mechanism}
Similarly, we show the differential of the $g-$function, namely the Abel differential $\d g$, which consists of two parts $$\d g=\frac{y}{4}\d g^{(y)}+2t\d g^{(t)}.$$ $\d g$ is corresponding to $\theta$ defined in \eqref{def theta} and $\Gamma$,  and is uniquely determined by
	\begin{align}\label{g condition 7}
	\left\{\begin{array}{ll}
		\d g^{(y)}\sim (1+\lambda^{-2})\d\lambda,& \lambda \rightarrow  \infty , 0 ,\\
		\oint_{\mathfrak{b}_j}\d g^{(y)}=0,& j=1,\cdots,4(p+q)-1.
	\end{array}\right.\ 
	\left\{\begin{array}{ll}
		\d g^{(t)}\sim -\frac{\lambda^4-6\lambda^2+1 }{4(1+\lambda^2)^3}\d\lambda,& \lambda \rightarrow   \pm\I,\\
		\oint_{\mathfrak{b}_j}\d g^{(t)}=0,& j=1,\cdots,4(p+q)-1.
	\end{array}\right.
\end{align}
		Let  $g-$function  $g(\lambda)=g(\lambda;y,t)$ be the Abel integral of $\mathrm{d}g$ begin with $(\lambda,R(\lambda))=(\I a_1^{-1},0)$ on $\mathcal{R}$. On the first sheet $g(\lambda)$ is well defined analytic function  on    $\lambda\in\mathbb{C}\backslash \Gamma$ which  satisfies:
		\begin{align}
			& g_+(\lambda) + g_-(\lambda) =  y\Omega_j^{(y)}+t\Omega_j^{(t)} & \lambda \in \Gamma_j,  \\
			& g(\lambda)-\theta(\lambda) =\mathcal{O}(1) & \lambda \rightarrow  \infty , 0 , \pm\I\\
			&\lim_{\lambda\to\infty}g(\lambda)-\theta(\lambda)=\frac{1}{4}( y\Omega_{4p+4q-1}^{(y)}+t\Omega_{4p+4q-1}^{(t)}),
		\end{align}	where the constants independent on $\lambda,\ y,\ t$ are given by $$\Omega_j^{(y)}=\left\{\begin{array}{lll}
			0& j=0,\\
			\displaystyle-\frac{1}{4}\sum_{l=1}^{j}\oint_{ \mathfrak{a}_l }\mathrm{d}g^{(y)},&j=1,\cdots,4p+4q-1,\end{array}\right.\quad
		\Omega_j^{(t)}=\left\{\begin{array}{lll}
			0,& j=0,\\
			\displaystyle-2\sum_{l=1}^{j}\oint_{ \mathfrak{a}_l }\mathrm{d}g^{(t)},&j=1,\cdots,4p+4q-1.\end{array}\right.
		$$

%
%
\subsection{Explicitly solvable RH problem}  Via the transformation  \begin{align}\label{trans1}
		M^{(1)}(\lambda):=\left\{\begin{array}{ll}
			\beta_{1}^{\frac{\hat{\sigma}_{3}}{2}}\mathrm{e}^{-\frac{ \I\pi}{4}\hat{\sigma}_{3}}\mathrm{e}^{-\frac{\I}{4}( y\Omega_{4p+4q-1}^{(y)}+t\Omega_{4p+4q-1}^{(t)})\sigma_{3} }M(\lambda)\mathrm{e}^{\I (g(\lambda)-\theta(\lambda))\sigma_{3}},\quad &q\neq0,\\
			\alpha_{p}^{\frac{\hat{\sigma}_{3}}{2}}\mathrm{e}^{-\frac{\I}{4}( y\Omega_{4p-1}^{(y)}+t\Omega_{4p-1}^{(t)})\sigma_{3} }M(\lambda)\mathrm{e}^{\I (g(\lambda)-\theta(\lambda))\sigma_{3}},\quad& q=0,
		\end{array}\right.
	\end{align} it is readily seen that 
   $M^{(1)}(\lambda)$ satisfies the following solvable RH problem.
	\begin{Rhp} \hfill \label{RHP g=7}
		\begin{itemize}
			\item  $M^{(1)}(\lambda)$ is holomorphic on $\lambda\in\mathbb{C}\backslash\Gamma$.
			\item $M^{(1)}_+(\lambda)=	M^{(1)}_-(\lambda)\left\{\begin{array}{ll}
			 \begin{pmatrix}
					0&\I  \\\I &0\\
				\end{pmatrix},\quad &\lambda\in\Gamma_0,\\[8pt]
				 \begin{pmatrix}
					0&\I\mathrm{e}^{-2\pi\I C_j}  \\\I\mathrm{e}^{2\pi \I C_j} &0\\
				\end{pmatrix},\quad &\lambda\in\Gamma_j,\ j=1,\cdots,4p+4q-1,\end{array}\right.$ \\
			 where $C_j$ are shown is \eqref{C qn0}-\eqref{C q=0} respectively.
			 \item $M^{(1)}(\lambda)$ has at most $-1/4$ singularity on the endpoints of $\Gamma$.
			\item As $\lambda\rightarrow\infty$, $ M^{(1)}(\lb)=I+\mathcal{O}(\lambda^{-1}).$
		\end{itemize}	\end{Rhp}
	
		For $q\neq0$, above constants  $C_j=C_j(y,t)$ are given by
		\begin{align}\label{C qn0}
			2\pi  C_j=\left\{\begin{array}{llllllll}
			y\Omega_j^{(y)}+t\Omega_j^{(t)}-\I \log\frac{\beta_{j+1}}{\beta_{1}},\quad & j=1,\cdots,q-1,\\
			y\Omega_j^{(y)}+t\Omega_j^{(t)}-\I \log\frac{\alpha_{p+q-j}}{\beta_{1}}+\frac{ \pi}{2} ,\quad & j=q,\cdots,p+q-1,\\
			y\Omega_j^{(y)}+t\Omega_j^{(t)}+\I \log\alpha_{j-p-q+1}\beta_{1}+\frac{ \pi}{2} ,\quad & j=p+q,\cdots,2p+q-1,\\
			y\Omega_j^{(y)}+t\Omega_j^{(t)}+\I\log\beta_{1}\beta_{2p+2q-j},\quad &  j=2p+q,\cdots,2p+2q-1,\\
			y\Omega_j^{(y)}+t\Omega_j^{(t)}-\I \log\frac{\beta_{j-2p-2q+1}}{\beta_{1}},\quad & j=2p+2q,\cdots,2p+3q-1,\\
			y\Omega_j^{(y)}+t\Omega_j^{(t)}-\I \log\frac{\alpha_{3p+3q-j}}{\beta_{1}}+\frac{ \pi}{2} ,\quad & j=2p+3q,\cdots,3p+3q-1,\\
			y\Omega_j^{(y)}+t\Omega_j^{(t)}+\I \log\alpha_{j-3p-3q+1}\beta_{1}+\frac{ \pi}{2} ,\quad & j=3p+3q,\cdots,4p+3q-1,\\
			y\Omega_j^{(y)}+t\Omega_j^{(t)}+\I\log\beta_{1}\beta_{4p+4q-j},\quad &  j=4p+3q,\cdots,4p+4q-1.
		\end{array}\right.
		\end{align}
 
		For $q=0$,  $C_j$ are given by
		\begin{align}\label{C q=0}
			2\pi  C_j=\left\{\begin{array}{llllllll}
			y\Omega_j^{(y)}+t\Omega_j^{(t)}-\I \log\frac{\alpha_{p-j}}{\alpha_{p}}+\frac{ \pi}{2} ,\quad & j=1,\cdots,p-1,\\
			y\Omega_j^{(y)}+t\Omega_j^{(t)}+\I \log\alpha_{j-p+1}\alpha_{p}+\frac{ \pi}{2} ,\quad & j=p,\cdots,2p-1,\\
			y\Omega_j^{(y)}+t\Omega_j^{(t)}-\I \log\frac{\alpha_{3p-j}}{\alpha_{p}}+\frac{ \pi}{2} ,\quad & j=2p,\cdots,3p-1,\\
			y\Omega_j^{(y)}+t\Omega_j^{(t)}+\I \log\alpha_{j-3p+1}\alpha_{p}+\frac{ \pi}{2} ,\quad & j=3p,\cdots,4p-1.
		\end{array}\right.
		\end{align} We do not need to be apprehensive about  multi-value of  the logarithm in \eqref{C qn0}-\eqref{C q=0} because any choice of it keeps the jump condition in the RH problem \ref{RHP g=7}. Let $\kappa=\kappa(\lambda)$ be
 \begin{align*}
		\kappa^4= \prod_{j=1}^{p}\frac{(\lambda^{2}+\I(a_{j}- a_{j}^{-1})\lambda+1)(\lambda^{2}-\I(b_{j}- b_{j}^{-1})\lambda+1)}{(\lambda^{2}-\I(a_{j}- a_{j}^{-1})\lambda+1)(\lambda^{2}+\I(b_{j}- b_{j}^{-1})\lambda+1)}
		\prod_{l=1}^{q}\frac{(\lambda^{2}-\mathrm{e}^{2\I c_{l}})(\lambda^{2}- \mathrm{e}^{-2\I d_{l}})}{(\lambda^{2}-\mathrm{e}^{-2\I c_{l}})(\lambda^{2}- \mathrm{e}^{2\I d_{l}})}.
	\end{align*} 
	which is analytic on $\lambda\in\mathbb{C}\backslash\Gamma$ and fixed by requiring $$  \kappa=1+\mathcal{O}(\lambda^{-1}), \lambda\rightarrow\infty.$$
	Denote the Abel map of $\mathcal{R}$ as  $\mathcal{A}(\lambda)$ on the first sheet:
	\begin{align*}
		\mathcal{A}(\lambda):=\left\{\begin{array}{ll}
			\left(\int_{\I a_{1}^{-1}}^{\lambda}\omega_j\right)_{j=1,\cdots,4p+4q-1}, \quad &q\neq0,\\
		 \left(\int_{-\E^{-\I c_{p}}}^{\lambda}\omega_j\right)_{j=1,\cdots,4p-1}, \quad &q=0,
		\end{array}\right.
	\end{align*} where $\omega_j$, $j=1,\cdots,4p+4q-1$ are the normalized holomorphic differential on $\mathcal{R}$ such that
	 \begin{align}
		\oint_{\mathfrak{b}_j}\omega_k=\delta_{jk}.
		\end{align}
	 The Riemann Theta function associated with period matrix $B$ $\Theta$ is given by \begin{align}
		\Theta(z):=\sum_{l\in\mathbb{Z}^{4p+4q-1}}\exp\left(\pi\I\left<l,Bl\right>+2\pi\I\left<l,z\right>\right), z\in\mathbb{C}^{4p+4q-1},
	\end{align} where $B=(B_{kj})_{j,k=1,\cdots,4p+4q-1}$ \begin{align}
	B_{kj}:=\sum_{l=1}^j\oint_{\mathfrak{a}_l}\omega_k,\quad j,k=1,\cdots,4p+4q-1.
	\end{align}
Let  $K:=(K_j)_{j=1,\cdots,4p+4q-1}\in\mathbb{C}^{4p+4q-1}$ be the Riemann constant of $\mathcal{R}$ with
	\begin{align}
		K_j=\frac{1}{2}\sum_{k=1}^{4p+4q-1}B_{kj}-\frac{j}{2}.
	\end{align}
	 Since   $\kappa^2-\kappa^{-2}:\mathcal{R}\to \mathbb{C}\cup\{\infty\}$ is a holomorphic mapping with degree $8(p+q)$ (where all branch points are simple poles), the zeros of $\kappa^2-\kappa^{-2}$ occur simultaneously on the upper and lower sheet of $\mathcal{R}$. Specifically,  if $(\lambda, R(\lambda))$ is a zero of $\kappa+\kappa^{-1}$ on the upper sheet,  then $(\lambda, -R(\lambda))$ is also a zero of $\kappa-\kappa^{-1}$ on the lower sheet. Introduce the  non-special divisor $\mathcal{D}$ as  the zero  of $\kappa+\kappa^{-1}$ on $\mathcal{R}$ except $(\infty,-R(\infty))$. 
	Thus for $$  e=\mathcal{A}(\mathcal{D})+K,$$ the function $
	\Theta(\mathcal{A}(P)-e)$ has precisely $4p+4q-1$ zeros corresponding to the divisor $\mathcal{D}$.
	The vector $C$ is determined by   \eqref{C qn0}-\eqref{C q=0}.
	With the help of above notations, 	the solution   $M^{(1)}(\lambda)$ of the RH problem \ref{RHP g=7}
	is explicitly  given by   \cite{Kot-BA}: 
	\begin{align} \label{M1 g}
		M^{(1)}(\lambda)=	\frac{1}{2}\begin{pmatrix}
			\frac{\Theta(\mathcal{A}(\infty)-e)}{\Theta(\mathcal{A}(\infty)+C-e)}&0\\
			0&\frac{\Theta(-\mathcal{A}(\infty)+e)}{\Theta(-\mathcal{A}(\infty)+C+e)}
		\end{pmatrix}
		\begin{pmatrix}
			\left(\kappa+\kappa^{-1}\right)\frac{\Theta(\mathcal{A}(\lambda)+C+e)}{\Theta(\mathcal{A}(\lambda)+e)}
			&	\left(\kappa-\kappa^{-1}\right)\frac{\Theta(-\mathcal{A}(\lambda)+C+e)}{\Theta(-\mathcal{A}(\lambda)+e)}\\
			\left(\kappa-\kappa^{-1}\right)\frac{\Theta(\mathcal{A}(\lambda)+C-e)}{\Theta(\mathcal{A}(\lambda)-e)}
			&	\left(\kappa+\kappa^{-1}\right)\frac{\Theta(-\mathcal{A}(\lambda)+C-e)}{\Theta(-\mathcal{A}(\lambda)-e)}
		\end{pmatrix}.\normalsize
	\end{align} 
	
We then obtain the exact expression of  $M(\lambda)$ by substituting \eqref{M1 g} into  \eqref{trans1}, then   the reconstruction formula \eqref{recons} shows  the precise genus-$(4p+4q-1)$ algebro-geometric solution depending on vector-valued parameters in \eqref{para1}-\eqref{para2} as \begin{align}\label{AGsol 7} 
		\begin{aligned}
			u^{(AG)}(y,t;\textbf{P}_1,\textbf{P}_2,\textbf{A},\textbf{B})&=\lim_{\lb\to \I}\frac{1}{\lb-\I}\left(1-\frac{m_1(\lb;y,t)m_2(\lb;y,t)}{m_1(\I;y,t)m_2(\I;y,t)}\right),
			\\
			x^{(AG)}(y,t;\textbf{P}_1,\textbf{P}_2,\textbf{A},\textbf{B})&=y+\log(\frac{m_1(\I;y,t)}{m_2(\I;y,t)})+2\I (yX^{(y)}+tX^{(t)}),
		\end{aligned}
	\end{align} where $X^{(y)}$, $X^{(t)}$ are constants  given by
	\begin{align}\label{X trans 7}
	X^{(y)}:=\frac{1}{4}\int_{\I}^{+\infty\I}\d g^{(y)}-(1+\lambda^{-2})\d\lambda,\quad X^{(t)}:=2\int_{\I}^{+\infty\I}\d g^{(t)}+\frac{\lambda^4-6\lambda^2+1 }{4(1+\lambda^2)^3}\d \lambda, 
	\end{align} with the path of integral   on the imaginary axis. It is well-defined because  that \eqref{g condition 7} implies that the multi-value of $\d g$ on imaginary axis can be ignored.
	In addition, { \footnotesize 
		\begin{align*}
			m_1(\lambda;y,t):=	  	\frac{(\kappa+\kappa^{-1})\Theta(\mathcal{A}(\infty)-e)\Theta(\mathcal{A}(\lambda)+C(y,t)+e)}{2\Theta(\mathcal{A}(\infty)+C(y,t)-e)\Theta(\mathcal{A}(\lambda)+e)}-\I\mathrm{e}^{-\I\pi C_{4p+4q-1}}	\frac{\left(\kappa-\kappa^{-1}\right)\Theta(-\mathcal{A}(\infty)+e)\Theta(\mathcal{A}(\lambda)+C(y,t)-e)}{2\Theta(-\mathcal{A}(\infty)+C(y,t)+e)\Theta(\mathcal{A}(\lambda)-e)} ,
		\end{align*}
		\begin{align*}
			m_2(\lambda;y,t):= 	\frac{(\kappa+\kappa^{-1})\Theta(-\mathcal{A}(\infty)+e)\Theta(-\mathcal{A}(\lambda)+C(y,t)-e)}{2\Theta(-\mathcal{A}(\infty)+C(y,t)+e)\Theta(-\mathcal{A}(\lambda)-e)}
			+ \I	\mathrm{e}^{\I\pi C_{4p+4q-1}}\frac{\left(\kappa-\kappa^{-1}\right)\Theta(\mathcal{A}(\infty)-e)\Theta(-\mathcal{A}(\lambda)+C(y,t)+e)}{2\Theta(\mathcal{A}(\infty)+C(y,t)-e)\Theta(-\mathcal{A}(\lambda)+e)}.
	\end{align*}}\normalsize

	\subsection{Alternative form of the solvable RH problem}
	Similarly, we have an alternative $g-$function  $\tilde{g}=\tilde{g}(\lambda)$ which is the Abel integral of 
	\begin{align}
		\d \tilde{g}:=	\frac{y}{4}\d g^{(y)}+2t\d g^{(t)},
	\end{align} from $(\I a^{-1},R(\I a^{-1}))\in\mathcal{R}$.  
$\tilde{g}^{(y)}$ and 
$\tilde{g}^{(t)}$are holomorphic differentials on $\mathcal{R}$ and  uniquely determined by 	\begin{align}
	\left\{\begin{array}{ll}
		\d \tilde{g}^{(y)}\sim (1+\lambda^{-2})\d\lambda,& \lambda \rightarrow  \infty , 0 ,\\
		\oint_{\mathfrak{a}_j}\d \tilde{g}^{(y)}=0,& j=1,\cdots,4p+4q-1.
	\end{array}\right.,\quad
	\left\{\begin{array}{ll}
		\d \tilde{g}^{(t)}\sim -\frac{\lambda^4-6\lambda^2+1 }{4(1+\lambda^2)^3}\d\lambda,& \lambda \rightarrow   \pm\I,\\
		\oint_{\mathfrak{a}_j}\d \tilde{g}^{(t)}=0,& j=1,\cdots,4p+4q-1.
	\end{array}\right.
\end{align}
  Thus the Abel integral $\tilde{g}(\lb)$ is analytic on $\C\backslash  \tilde{\Gamma }$, where $\tilde{\Gamma}=\bigcup_{j=1}^{4p+4q-1}\gamma_j\cup\Gamma $, $\gamma_j$ denotes the part of $\mathfrak{a}_j$ on the first sheet of $\mathcal{R}$ (as Figure \ref{gamma-7}), and  
	\begin{align*}
		& \tilde{g}_+(\lb) - \tilde{g}_-(\lb) =  y\tilde{\Omega}_j^{(y)}+t\tilde{\Omega}_j^{(t)} & \lambda \in \gamma_j,  \\
		& \tilde{g}_+(\lb) +\tilde{g}_-(\lb) = 0& \lambda \in \Gamma_j,  \\
		& \tilde{g}(\lb)-\theta(\lb) =\mathcal{O}(1) & \lambda \rightarrow  \infty , 0 , \pm\I,\\	&\lim_{\lambda\to\infty}(\tilde{g}(\lambda)-\theta(\lambda))=0.
	\end{align*} where 
	$$\tilde{\Omega}_j^{(y)}=\frac{1}{4}\sum_{l=j}^{4p+4q-1}\oint_{ \mathfrak{b}_l }\mathrm{d}\tilde{g}^{(y)},\quad
	\tilde{\Omega}_j^{(t)}= 2\sum_{l=j}^{4p+4q-1}\oint_{ \mathfrak{b}_l }\mathrm{d}\tilde{g}^{(t)},\quad j=1,\cdots,4p+4q-1.
	$$

	The new $g-$function $\tilde{g}(\lambda)$ admits a new transformation from the RH problem \ref{RHP0 g}. Let \begin{align}\label{trans til}
		 \tilde{M}^{(1)}(\lb):=M(\lambda)\mathrm{e}^{\I (\tilde{g}(\lambda)-\theta(\lambda))\sigma_{3}},
	\end{align} we have that $\tilde{M}^{(1)}(\lambda)$ satisfies the following RH problem. 
	 	\begin{figure}
	 	[h]
	 	\begin{tikzpicture} [xscale=1,yscale=1]

	 		\draw[rounded corners,thin](-1,1.732)arc(120:105:2)node[below]{$\gamma_{q}$}arc(105:90:2)--(0,2.8);
	 		\draw[thin,->](-1,1.732)arc(120:105:2);  
	 		
	 		\draw[thin,->](-1.8126,-0.8452)arc(205:180:2)node[left]{$\gamma_{p+q}$};
	 		\draw[thin](-2,0)arc(180:155:2);
	 		
	 		\draw[rounded corners,thin](-1 ,-1.732)arc(240:255:2)arc(255:270:2)--(0,-1.2) (-1,-1.82)node[below]{$\gamma_{2p+q}$};
	 		\draw[thin,-<](-1 ,-1.732)arc(240:255:2);

	 		\draw[rotate=180,rounded corners,thin](-1,1.732)arc(120:105:2)node[below]{$\gamma_{4p+3q}$}arc(105:90:2)--(0,2.8);
	 		\draw[rotate=180,thin,-<](-1,1.732)arc(120:105:2);  
	 		
	 		\draw[rotate=180,thin,-<-=.5](-1.8126,-0.8452)arc(205:155:2) (-2,0)node[right]{$\gamma_{3p+3q}$};

	 		\draw[rotate=180,rounded corners,thin](-1 ,-1.732)arc(240:255:2)arc(255:270:2)--(0,-1.2)
	 		(-1,-1.82)node[above]{$\gamma_{2p+3q}$};
	 		\draw[rotate=180,thin,->](-1 ,-1.732)arc(240:255:2);

	 		\draw[thin,->](0,0.8)--(0,0)node[left]{$\gamma_{2p+2q}$};
	 		\draw[thin](0,0.8)--(0,-0.8);
	 		
	 		\draw[thin,->-=.5](0,2.9)--(0,3.3) (0,3.1)node[left]{$\gamma_{1}$};
	 		\draw[thin,-<-=.5](0,-2.9)--(0,-3.3) (0,-3.1)node[left]{$\gamma_{4p+4q-1}$};

	 		\fill[rotate=-55,white] (-0.2,1.9)--(0.2,1.9)--(0.2,2.1)--(-0.2,2.1); \draw[rotate=-55](0,2)node{\normalcolor$\cdots$};
	 		\draw[rotate=-30,->-=.6,thick,red] (0,2)arc(90:75:2);
	 		\draw[rotate=-65,->-=.6,thick,red] (0,2)arc(90:75:2);
	 		\draw[red] (1.3,1.7)node[right]{$\Gamma_{2p+3q}$};
	 		\draw[red] (1.9,.6)node[right]{$\Gamma_{3p+3q-1}$};

	 		\fill[rotate=55,white] (-0.25,1.9)--(0.25,1.9)--(0.25,2.1)--(-0.25,2.1); \draw[rotate=55](0,2)node{\normalcolor$\cdots$};
	 		\draw[rotate=45,->-=.6,thick,red] (0,2)arc(90:75:2);
	 		\draw[rotate=80,->-=.6,thick,red] (0,2)arc(90:75:2);
	 		\draw[red] (-1.3,1.7)node[left]{$\Gamma_{q}$};
	 		\draw[red] (-1.9,.6)node[left]{$\Gamma_{p+q-1}$};
	 		
	 		\fill[rotate=-125,white] (-0.25,1.9)--(0.25,1.9)--(0.25,2.1)--(-0.25,2.1); \draw[rotate=-125](0,2)node{\normalcolor$\cdots$};
	 		\draw[rotate=-135,->-=.6,thick,red] (0,2)arc(90:75:2);
	 		\draw[rotate=-100,->-=.6,thick,red] (0,2)arc(90:75:2);
	 		\draw[red] (1.3,-1.7)node[right]{$\Gamma_{4p+3q-1}$};
	 		\draw[red] (1.9,-.6)node[right]{$\Gamma_{3p+3q}$};

	 		\fill[rotate=125,white] (-0.25,1.9)--(0.25,1.9)--(0.25,2.1)--(-0.25,2.1); \draw[rotate=125](0,2)node{$\cdots$};
	 		\draw[rotate=150,->-=.6,thick,red] (0,2)arc(90:75:2);
	 		\draw[rotate=115,->-=.6,thick,red] (0,2)arc(90:75:2);
	 		\draw[red] (-1.3,-1.7)node[left]{$\Gamma_{2p+q-1}$};
	 		\draw[red] (-1.9,-.6)node[left]{$\Gamma_{p+q}$};
	 		
	 		\draw[->-=.6,thick,red] (0,0.4)--(0,0.8);
	 		\fill[white] (0.1,0.9)--(0.1,1.3)--(-0.1,1.3)--(-0.1,0.9);
	 		\draw (0,1.1)node{\normalcolor$\cdots$};
	 		\draw[->-=.6,thick,red] (0,1.4)--(0,1.8);
	 		\draw[red] (0.1,1.3)node[right]{ $\Gamma_{2p+3q-1}$};
	 		\draw[red] (0.1,.6)node[right]{ $\Gamma_{2p+2q}$};
	 		
	 		\draw[-<-=.6,thick,red] (0,-0.4)--(0,-0.8);
	 		\fill[white] (0.1,-0.9)--(0.1,-1.3)--(-0.1,-1.3)--(-0.1,-0.9);
	 		\draw (0,-1.1)node{\normalcolor$\cdots$};
	 		\draw[-<-=.6,thick,red] (0,-1.4)--(0,-1.8);
	 		\draw[red] (0.1,-1.3)node[right]{ $\Gamma_{2p+q}$};
	 		\draw[red] (0.1,-.6)node[right]{ $\Gamma_{2p+2q-1}$};
	 		
	 		\draw[-<-=.6,thick,red] (0,3.7)--(0,3.3);
	 		\fill[white] (0.1,2.9)--(0.1,2.7)--(-0.1,2.7)--(-0.1,2.9);
	 		\draw (0,2.8)node{\normalcolor$\cdots$};
	 		\draw[-<-=.6,thick,red] (0,2.6)--(0,2.2);
	 		\draw[red] (0.1,3.5)node[right]{ $\Gamma_{0}$};
	 		\draw[red] (-0.1,2.4)node[left]{ $\Gamma_{q-1}$};
	 		
	 		\draw[->-=.6,thick,red] (0,-3.7)--(0,-3.3);
	 		\fill[white] (0.1,-2.9)--(0.1,-2.7)--(-0.1,-2.7)--(-0.1,-2.9);
	 		\draw (0,-2.8)node{\normalcolor$\cdots$};
	 		\draw[->-=.6,thick,red] (0,-2.6)--(0,-2.2);
	 		\draw[red] (0.1,-3.5)node[right]{ $\Gamma_{4p+4q-1}$};
	 		\draw[red] (-0.1,-2.4)node[left]{ $\Gamma_{4p+3q}$};

	 	\end{tikzpicture}
	 	\caption{Jump curve $\tilde{\Gamma}$ of the RH problem \ref{RHP g=7 alt}.}\label{gamma-7}
	 \end{figure}
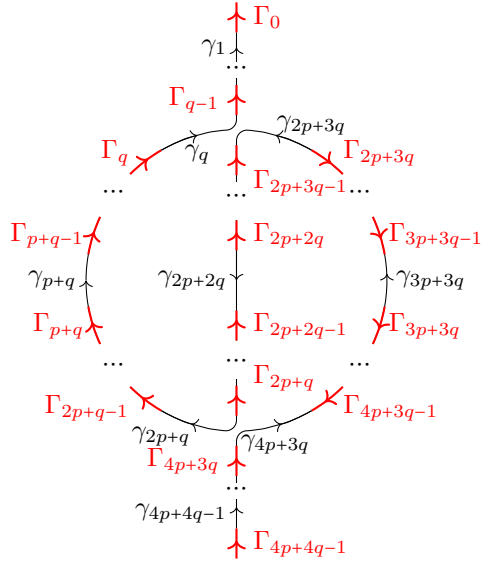
	\begin{Rhp}  \hfill\label{RHP g=7 alt}
		\begin{itemize}
			\item $\tilde{M}^{(1)}(\lambda)$ is holomorphic on $\lambda\in\mathbb{C}\backslash\tilde{\Gamma}$.
			\item$\tilde{M}_+^{(1)}(\lambda) =\tilde{M}_-^{(1)}(\lambda)\left\{\begin{array}{llll}
					\begin{pmatrix}
						0&-\beta_j^{-1} \\\beta_j&0\\
					\end{pmatrix},\quad \lambda\in\Gamma_{j-1}\cup\Gamma_{j+2p+2q-1},j=1,\cdots,q;\\
					\begin{pmatrix}
						0&\I\alpha_j^{-1} \\\I\alpha_j&0\\
					\end{pmatrix}, \quad \lambda\in\Gamma_{p+q-j}\cup\Gamma_{3p+3q-j},j=1,\cdots,p;\\
					\begin{pmatrix}
						0&\I\alpha_j \\\I\alpha_j^{-1}&0\\
					\end{pmatrix},\quad \lambda\in\Gamma_{j+p+q-1}\cup\Gamma_{j+3p+3q-1},j=1,\cdots,p;\\
					\begin{pmatrix}
						0&-\beta_j \\\beta_j^{-1}&0\\
					\end{pmatrix},\quad \lambda\in\Gamma_{2p+2q-j}\cup\Gamma_{4p+4q-j},j=1,\cdots,q;\\
					\mathrm{e}^{\I (y\tilde{\Omega}_j^{(y)}+t\tilde{\Omega}_j^{(t)})\sigma_{3}},\quad 	\lambda\in\gamma_j,\ j=1,\cdots,4p+4q-1.
					\end{array}\right.$
			
				\item $\tilde{M}^{(1)}(\lambda)$ has at most $-1/4$ singularity on the endpoints of $\Gamma$.
				\item $ \tilde{M}^{(1)}(\lb)=I+\mathcal{O}(\lambda^{-1}),\quad\lambda\rightarrow\infty.$
		\end{itemize}
	\end{Rhp}
Using the  transformation \eqref{trans til} from $M$ to $\tilde{M}^{(1)}$, the RH problem \ref{RHP g=7 alt} is associated with the solution \eqref{AGsol 7} via  reconstruction formula \eqref{recons}.  
	\begin{corollary}
		Denote $			\tilde{u}(x(y,t)),  \ \tilde{x}(y,t)$ as the result of substituting $\tilde{M}^{(1)}(\lambda)$      into the reconstruction formula \eqref{recons}, one can find that {\rm\begin{align}
		\begin{aligned}
			&\tilde{u}(x(y,t),t)=u^{(AG)}(y,t;\textbf{P}_1,\textbf{P}_2,\textbf{A},\textbf{B}),\\
			&\tilde{x}(y,t)= x^{(AG)}(y,t;\textbf{P}_1,\textbf{P}_2,\textbf{A},\textbf{B})+2\I \left(y\tilde{X}^{(y)}+t\tilde{X}^{(t)}\right),
		\end{aligned}
		\end{align}}
		where \begin{align*}
			\tilde{X}^{(y)}=\frac{1}{4}\lim_{\lambda\to\I}(\tilde{g}^{(y)}(\lambda)-(\lb-\lb^{-1})),\quad	\tilde{X}^{(t)}=2\lim_{\lambda\to\I}(\tilde{g}^{(t)}(\lambda)+\frac{\lb-\lb^{-1}}{(\lambda+\lambda^{-1})^2}).\end{align*}
	\end{corollary}
%
	
\vspace{3mm}

\noindent{\bf Acknowledgements.}
 The work of Fan is partially
supported by NSFC under grants 12271104, 51879045.  The work of Yang is partially supported by
Innovation Center for Mathematical Analysis of Fluid and Chemotaxis (Chongqing University).

\bibliographystyle{alpha}
\bibliography{ref-BA}

\newcommand{\etalchar}[1]{$^{#1}$}
\begin{thebibliography}{GGJM21}

\bibitem[BV07]{Buc2007}
R.~Buckingham and S.~Venakides.
\newblock Long-time asymptotics of the nonlinear schr\"odinger equation shock
  problem.
\newblock {\em Commun. Pure Appl. Math.}, 60(9):1349--1414, 2007.

\bibitem[CG90]{Cao1990}
C.~Cao and X.~Geng.
\newblock Neumann and bargmann systems associated with the coupled {K}d{V}
  soliton hierarchy.
\newblock {\em J. Phys. A Gen. Phys.}, 23(18):4117--4125, 1990.

\bibitem[DN74]{Dub1974}
B.~Dubrovin and S.~Novikov.
\newblock The periodic problem for the {K}orteweg--de {V}ries and
  {S}turm--{L}iouville equations. {T}heir connection with algebraic geometry.
\newblock {\em Dokl. akad. nauk Sssr}, 1974.

\bibitem[DVZ94]{DVZ1994}
P.~Deift, S.~Venakides, and X.~Zhou.
\newblock The collisionless shock region for the long‐time behavior of
  solutions of the kdv equation.
\newblock {\em Commun. Pure Appl. Math.}, 47(2):199--206, 1994.

\bibitem[DZ93]{Deift1993}
P.~Deift and X.~Zhou.
\newblock A steepest descent method for oscillatory {R}iemann--{H}ilbert
  problems. asymptotics for the mkdv equation.
\newblock {\em Ann. Math.}, 137(2):295--368, 1993.

\bibitem[DZZ16]{DZZ16}
S.~Dyachenko, D.~Zakharov, and V.~Zakharov.
\newblock Primitive potentials and bounded solutions of the kdv equation.
\newblock {\em Phys. D}, 333:148--156, 2016.

\bibitem[EGH{\etalchar{+}}17]{Ges2017}
J.~Eckhardt, F.~Gesztesy, H.~Holden, A.~Kostenko, and G.~Teschl.
\newblock Real-valued algebro-geometric solutions of the two-component
  {C}amassa-{H}olm hierarchy.
\newblock {\em Ann. Inst. Fourier (Grenoble)}, 67(3):1185--1230, 2017.

\bibitem[EGKT13]{Iry2013}
I.~Egorova, Z.~Gladka, V.~Kotlyarov, and G.~Teschl.
\newblock Long-time asymptotics for the korteweg–de vries equation with
  step-like initial data.
\newblock {\em Nonlinearity}, 26(7):1839--1864, 2013.

\bibitem[EMT18]{Ego2018}
I.~Egorova, J.~Michor, and G.~Teschl.
\newblock Long-time asymptotics for the toda shock problem: Non-overlapping
  spectra.
\newblock {\em J. Math. Phys. Anal. Geo.}, 14(4), 2018.

\bibitem[FLT20]{Feng2020}
B.~Feng, L.~Ling, and D.~Takahashi.
\newblock Multi-breather and high-order rogue waves for the nonlinear
  schr\"odinger equation on the elliptic function background.
\newblock {\em Stud. Appl. Math.}, 144(1):46--101, 2020.

\bibitem[FLY25]{mCHgas}
E.~Fan, G.~Li, and Y.~Yang.
\newblock Soliton gas for the modified camassa-holm equation and its
  asymptotics.
\newblock {\em submitted}, 2025.

\bibitem[Fok95]{1995On}
A.~S. Fokas.
\newblock On a class of physically important integrable equations.
\newblock {\em Phys. D}, 87(1-4):145--150, 1995.

\bibitem[Fuc96]{1996Some}
B.~Fuchssteiner.
\newblock Some tricks from the symmetry-toolbox for nonlinear equations:
  Generalizations of the camassa-holm equation.
\newblock {\em Phys. D}, 95(3-4):229--243, 1996.

\bibitem[GGH05]{Ges2005}
J.~Geronimo, F.~Gesztesy, and H.~Holden.
\newblock Algebro-geometric solutions of the {B}axter–{S}zeg difference
  equation.
\newblock {\em Commun. Math. Phys.}, 258(1):149--177, 2005.

\bibitem[GGJM21]{Gir2021}
M.~Girotti, T.~Grava, R.~Jenkins, and K.~D. T.-R. Mclaughlin.
\newblock Rigorous asymptotics of a kdv soliton gas.
\newblock {\em Commun. Math. Phys.}, 384(2):733--784, 2021.

\bibitem[GH03a]{Ges2003}
F.~Gesztesy and H.~Holden.
\newblock Algebro-geometric solutions of the {C}amassa-{H}olm hierarchy.
\newblock {\em Rev. Mat. Iberoamericana}, 19(1):73--142, 2003.

\bibitem[GH03b]{GesBookI}
F.~Gesztesy and H.~Holden.
\newblock {\em Soliton equations and their algebro-geometric solutions. Vol.
  I}, volume~79 of {\em Cambridge Studies in Advanced Mathematics}.
\newblock Cambridge University Press, Cambridge, 2003.

\bibitem[GH08]{Ges2008}
F.~Gesztesy and H.~Holden.
\newblock Real-valued algebro-geometric solutions of the {C}amassa-{H}olm
  hierarchy.
\newblock {\em Philos. Trans. R. Soc. Lond. Ser. A Math. Phys. Eng. Sci.},
  366(1867):1025--1054, 2008.

\bibitem[GWC99]{Geng1999}
X.~Geng, Y.~Wu, and C.~Cao.
\newblock Quasi-periodic solutions of the modified {K}adomtsev-{P}etviashvili
  equation.
\newblock {\em J. Phys. A Gen. Phys.}, 32(20):3733, 1999.

\bibitem[HFQ17]{Hou2017}
Y.~Hou, E.~Fan, and Z.~Qiao.
\newblock The algebro-geometric solutions for the
  {F}okas-{O}lver-{R}osenau-{Q}iao ({FORQ}) hierarchy.
\newblock {\em J. Geom. Phys.}, 117:105--133, 2017.

\bibitem[IM75]{Its1975}
A.~Its and V.~Matveev.
\newblock {S}chr\"odinger operators with finite-gap spectrum and {N}-soliton
  solutions of the {K}orteweg-de-{V}ries equation.
\newblock {\em Theor. Math. Phys}, 23:51--68, 1975.

\bibitem[KM19]{Kot2019}
V.~Kotlyarov and A.~Minakov.
\newblock Dispersive shock wave, generalized laguerre polynomials, and
  asymptotic solitons of the focusing nonlinear schr\"odinger equation.
\newblock {\em J. Math. Phys.}, 60(12):123501, 31, 2019.

\bibitem[Kot18]{Kot2018}
V.~Kotlyarov.
\newblock A matrix baker-akhiezer function associated with the maxwell-bloch
  equations and their finite-gap solutions.
\newblock {\em Symmetry Integr. Geom.}, 14:Paper No. 082, 27, 2018.

\bibitem[KS17]{Kot-BA}
V.~Kotlyarov and D.~Shepelsky.
\newblock Planar unimodular baker-akhiezer function for the nonlinear
  schrdinger equation.
\newblock {\em Ann. Math. Sci. Appl.}, 2(2):343--384, 2017.

\bibitem[KT07]{Kam2007}
S.~Kamvissis and G.~Teschl.
\newblock Stability of periodic soliton equations under short range
  perturbations.
\newblock {\em Phys. Lett. A}, 364(6):480--483, 2007.

\bibitem[KT09]{Kru2009}
H.~Kr\"uger and G.~Teschl.
\newblock Stability of the periodic toda lattice in the soliton region.
\newblock {\em Int. Math. Res. Not}, 2009.

\bibitem[MLT12]{Mik2012}
A.~Mikikits-Leitner and G.~Teschl.
\newblock Long-time asymptotics of perturbed finite-gap korteweg-de vries
  solutions.
\newblock {\em J. d'Analyse Math}, 116(1):163--218, 2012.

\bibitem[OR96]{Peter1996Tri}
P.~Olver and P.~Rosenau.
\newblock Tri-hamiltonian duality between solitons and solitary-wave solutions
  having compact support.
\newblock {\em Phys. Rev. E}, 53(2):1900–1906, 1996.

\bibitem[Sch96]{Sch1996}
J.~Schiff.
\newblock Zero curvature formulations of dual hierarchies.
\newblock {\em J. Math. Phys.}, 37(4):1928--1938, 1996.

\bibitem[SKBP24]{She2024}
D.~Shepelsky, I.~Karpenko, S.~Bogdanov, and J.~Prilepsky.
\newblock Periodic finite-band solutions to the focusing nonlinear
  {S}chr\"odinger equation by the {F}okas method: inverse and direct problems.
\newblock {\em Proc. A.}, 480(2286):20230828, 28, 2024.

\bibitem[YF22]{Yang2022}
Y.~Yang and E.~Fan.
\newblock On the long-time asymptotics of the modified camassa-holm equation in
  space-time solitonic regions.
\newblock {\em Adv. Math.}, 402:Paper No. 108340, 78, 2022.

\bibitem[ZF20]{Zhao2020}
P.~Zhao and E.~Fan.
\newblock Finite gap integration of the derivative nonlinear {S}chr\"odinger
  equation: a {R}iemann-{H}ilbert method.
\newblock {\em Phys. D}, 402:132213, 31, 2020.

\bibitem[ZF23]{Zhao2023}
P.~Zhao and E.~Fan.
\newblock A {R}iemann-{H}ilbert method to algebro-geometric solutions of the
  {K}orteweg--de {V}ries equation.
\newblock {\em Phys. D}, 454:133879, 24, 2023.

\end{thebibliography}
\end{document}